\documentclass[prodmode,acmtap]{acmlarge} 

\usepackage[ruled]{algorithm2e}

\SetAlFnt{\small}
\SetAlCapFnt{\small}
\SetAlCapNameFnt{\small}
\SetAlCapHSkip{0pt}
\IncMargin{-\parindent}

\usepackage{graphicx}
\usepackage{balance}
\usepackage{comment} 

\usepackage{ifpdf}

\usepackage{makeidx}  
\usepackage{amssymb,latexsym,amsfonts,amsmath,stmaryrd}
\usepackage{graphicx}

\usepackage{graphics} 
\usepackage{epsfig} 
\usepackage{mathptmx} 
\usepackage{amsmath} 
\usepackage{amssymb}  
\usepackage{listings}
\usepackage{color}
\usepackage{url}
\usepackage{graphicx}
\usepackage{subfig}
\usepackage{caption}
\DeclareCaptionType{copyrightbox}
\usepackage{upgreek}
\usepackage{siunitx}
\usepackage{varwidth}
\usepackage{tikz}
\usetikzlibrary{shapes,shadows,calc,snakes}
\usetikzlibrary{decorations}
\usepackage{pgfplots}

\usepackage{epstopdf}
\usepgflibrary{arrows}
\usepackage{circuitikz}
\usepackage{makecell}

\usepackage{sidecap}

\usepackage{bbm}

\newcommand{\R}{{\mathbb{R}}}

\newcommand{\Ce}{{\mathbb C}}
\newcommand{\N}{{\mathbb{N}}}

\newcommand{\SNR}{\mathsf{SNR}}

\newtheorem{lem}{Lemma}
\newtheorem{cor}{Corollary}

\newcommand{\Ga}{\Gamma}
\newcommand{\Al}{\alpha}
\newcommand{\eps}{\epsilon}

\newcommand{\pa}[1]{\left( #1 \right)}
\newcommand{\pac}[1]{\left\{ {#1} \right\}}
\newcommand{\beq}[2]{\begin{equation}\label{#1} #2 \end{equation}}
\newcommand{\bal}[2]{{\setlength\arraycolsep{2pt}\begin{eqnarray}\label{#1} #2 \end{eqnarray}}}
\newcommand{\beqn}[2]{\begin{equation*}\label{#1} {#2} \end{equation*}}
\newcommand{\baln}[2]{\begin{align*}\label{{#1}} {#2} \end{align}}

\newcommand{\nn}{\nonumber}

\newcommand{\sinewave}[4]{%
		(#1,0+#3) 
		sin  (#1+#2,#4+#3) 
		cos (#1+#2+#2,0+#3) 
		sin  (#1+#2+#2+#2,-#4+#3) 
		cos (#1+#2+#2+#2+#2,0+#3) 
		sin  (#1+#2+#2+#2+#2+#2,#4+#3) 
		cos (#1+#2+#2+#2+#2+#2+#2,0+#3) 
		sin  (#1+#2+#2+#2+#2+#2+#2+#2,-#4+#3) 
		cos (#1+#2+#2+#2+#2+#2+#2+#2+#2,0+#3) 
		sin  (#1+#2+#2+#2+#2+#2+#2+#2+#2+#2,#4+#3) 
		cos (#1+#2+#2+#2+#2+#2+#2+#2+#2+#2+#2,0+#3) 
		sin  (#1+#2+#2+#2+#2+#2+#2+#2+#2+#2+#2+#2,-#4+#3)
		cos (#1+#2+#2+#2+#2+#2+#2+#2+#2+#2+#2+#2+#2,0+#3)
		sin  (#1+#2+#2+#2+#2+#2+#2+#2+#2+#2+#2+#2+#2+#2,#4+#3)
		cos (#1+#2+#2+#2+#2+#2+#2+#2+#2+#2+#2+#2+#2+#2+#2,0+#3)
}

\newcommand{\halfsinewave}[4]{%
		(#1,0+#3) 
		sin  (#1+#2,#4+#3) 
		cos (#1+#2+#2,0+#3) 
		sin  (#1+#2+#2+#2,-#4+#3) 
		cos (#1+#2+#2+#2+#2,0+#3) 
		sin  (#1+#2+#2+#2+#2+#2,#4+#3) 
		cos (#1+#2+#2+#2+#2+#2+#2,0+#3) 
}

\newcommand{\attackwave}[5]{%
		(#1,0+#3) 
		sin  (#1+#2,                    								#3+#4) 
		cos (#1+#2+#2,              								#3+#4-#5-#5)
		cos (#1+#2+#2+#2,        								#3+#4+#5+#5)
		cos (#1+#2+#2+#2+#2,  								#3+#4-#5-#5)
		cos (#1+#2+#2+#2+#2+#2,  							#3+#4)
		cos (#1+#2+#2+#2+#2+#2+#2,  						#3)
		sin  (#1+#2+#2+#2+#2+#2+#2+#2,                    				#3-#4) 
		cos (#1+#2+#2+#2+#2+#2+#2+#2+#2,              				#3-#4+#5+#5)
		cos (#1+#2+#2+#2+#2+#2+#2+#2+#2+#2,        			#3-#4-#5-#5)
		cos (#1+#2+#2+#2+#2+#2+#2+#2+#2+#2+#2,  			#3-#4+#5+#5)
		cos (#1+#2+#2+#2+#2+#2+#2+#2+#2+#2+#2+#2,  			#3-#4)
		cos (#1+#2+#2+#2+#2+#2+#2+#2+#2+#2+#2+#2+#2,  		#3)
		sin  (#1+#2+#2+#2+#2+#2+#2+#2+#2+#2+#2+#2+#2+#2,                    								#3+#4) 
		cos (#1+#2+#2+#2+#2+#2+#2+#2+#2+#2+#2+#2+#2+#2+#2,              								#3+#4-#5-#5)
		cos (#1+#2+#2+#2+#2+#2+#2+#2+#2+#2+#2+#2+#2+#2+#2+#2,        								#3+#4+#5+#5)
		cos (#1+#2+#2+#2+#2+#2+#2+#2+#2+#2+#2+#2+#2+#2+#2+#2+#2,  								#3+#4-#5-#5)
		cos (#1+#2+#2+#2+#2+#2+#2+#2+#2+#2+#2+#2+#2+#2+#2+#2+#2+#2,  							#3+#4)
		cos (#1+#2+#2+#2+#2+#2+#2+#2+#2+#2+#2+#2+#2+#2+#2+#2+#2+#2+#2,  							#3)
}

\acmVolume{2}
\acmNumber{3}
\acmArticle{1}
\articleSeq{1}
\acmYear{2010}
\acmMonth{5}

\makeatletter
\def\maketitle{\newpage \thispagestyle{titlepage}\par
  \begingroup \lineskip = \z@\null \vskip -13.5pt\relax 
  \parindent\z@ {\hyphenpenalty\@M
    {\titlefont \@title \par
    \global\firstfoot
    \global\runningfoot
  }}
  \global\@firstpg\the\c@page
      {\vskip 13.5pt\relax \normalsize \authorfont 
    \begingroup \addtolength{\baselineskip}{2pt}
    \@author\par \vskip -2pt 
    \endgroup }
      {
    \baselineskip 17pt\relax
    \hbox{\vrule height .2pt width \@acmWidth}
      }
      \vskip 8.5pt \footnotesize \box\@abstract \vskip 4pt\relax 
         {\def\and{\unskip\/{\rm ; }}
}\par\vskip 4pt\relax
         \ifx\@acmformat\@empty\else
             \footnotesize \hsize \@acmWidth \parindent 0pt \noindent
             \vskip 4\p@
             \noindent  {\bf ACM Reference Format:}\\[2pt]
             \@acmformat\vskip 0.5\p@
             \par\fi%
         {\baselineskip 14pt\relax
           \@abstractbottom
         }
         \vskip 23pt\relax
         \endgroup
         \let\maketitle\relax
         \gdef\@categories{}}
\newbox\@abstract
\newbox\@terms
\newbox\@keywords
\def\abstract{\global\setbox\@abstract=\vbox\bgroup \everypar{}
  \footnotesize \hsize \@acmWidth \parindent 10pt \noindent 
  \rule{0pt}{10pt}\ignorespaces}
\def\endabstract{\egroup}
\makeatother

\title{Attack Resilience and Recovery using Physical Challenge Response Authentication for Active Sensors Under Integrity Attacks}

\author{Yasser Shoukry
\affil{University of California at Los Angeles, USA}
Paul Martin
\affil{University of California at Los Angeles, USA}
Yair Yona
\affil{University of California at Los Angeles, USA}
Suhas Diggavi
\affil{University of California at Los Angeles, USA}
Mani Srivastava
\affil{University of California at Los Angeles, USA}
}

\begin{abstract}
Embedded sensing systems are pervasively used in life- and security-critical  systems such as those found in airplanes, automobiles, and healthcare. 
Traditional security mechanisms for these sensors focus on data encryption and other post-processing techniques, but the sensors themselves often remain vulnerable to attacks in the  physical/analog domain.  If an adversary manipulates a physical/analog signal prior to digitization, no amount of digital security mechanisms after the fact can help.  Fortunately, nature imposes fundamental constraints on how these analog signals can behave. This work presents PyCRA, a physical challenge-response authentication scheme designed to protect active sensing systems against \emph{physical} attacks occurring in the analog domain. PyCRA provides security for active sensors by continually \emph{challenging} the surrounding environment via random but deliberate physical probes.  By analyzing the responses to these probes, and by using the fact that the adversary cannot change the underlying laws of physics, we provide an authentication mechanism that not only detects malicious attacks but provides resilience against them.  We demonstrate the effectiveness of PyCRA in detecting and mitigating attacks through several case studies using two sensing systems:  (1) magnetic sensors like those found on gear and wheel speed sensors in robotics and automotive, and (2) commercial Radio Frequency Identification (RFID) tags used in many security-critical applications. Finally, we outline methods and theoretical proofs for further enhancing the resilience of PyCRA to active attacks by means of a \emph{confusion phase}---a period of low signal to noise ratio that makes it more difficult for an attacker to correctly identify and respond to PyCRA's physical challenges. In doing so, we evaluate both the robustness and the limitations of the PyCRA security scheme, concluding by outlining practical considerations as well as further applications for the proposed authentication mechanism.
\end{abstract}




\begin{document}

\maketitle
\section{Introduction}
\label{sec:introduction}

Recent decades have witnessed a proliferation in embedded sensors for observing a variety of physical phenomena.  Increased use of these sensors in security- and life-critical applications has been accompanied by a corresponding increase in attacks targeting sensing software, hardware, and even physical, analog signals themselves. While considerable research has explored sensor security from a system-level perspective---network redundancy, sensor fusion, ... etc---sensors themselves remain largely vulnerable to attacks targeting analog signals prior to digitization. This vulnerability can lead to catastrophic failures when a malicious third party attempts to spoof the sensor \cite{ghosttalk,Willams_SCADA,Sastry_HOTSEC08,YasserABS,AttackGryroUsenix}.

Several \emph{system-level} sensor security schemes have been proposed in the context of power grids.  For example, Dorfler et al. have explored distributed cyber-physical attack detection in the context of power networks \cite{Bullo_TAC}. Similar ideas for providing system-level security in smart grids can be found in \cite{KimPowerAttack,KosutPowerAttack,SandbergPowerAttack,LiuPowerAttack,KalleGrid}. 
Security schemes in this vein include, among others, state-space and control-theoretic approaches for detecting anomalous system behavior \cite{Hamza_TAC,Yasser_SMT,Bullo_TAC}.  One idea common to these efforts is that an inherent security mechanism and robustness can be found in the physics governing the dynamics of the \emph{system} as a whole.  For example, a mismatch between the rate of change in a vehicle's location as reported by GPS and by the odometer sensor may indicate that one of these two sensors is either faulty or under attack.

A complementary security mechanism can be found in the physics governing the \emph{sensor} itself.  If a sensor observes an analog signal that appears to violate the physics governing the sensing dynamics, the signal itself may be under attack, necessitating security mechanisms at the analog signal level.  To reduce sensor-level vulnerabilities, engineers often place sensors in secure or remote physical locations to preclude direct physical contact with the sensing hardware. Additionally, the phenomenon being sensed is often difficult to access, whether prohibitively far away or surrounded by protective material. In such scenarios, adversaries have access only to the analog signal prior to it reaching the sensor, and their attack must be carried out without direct access to any hardware in the entire sensing path, from source to sink. Even with these countermeasures in place, an adversary can still attack sensors by manipulating the physical signals before their transduction and subsequent digitization  \cite{ghosttalk,YasserABS}. Robust countermeasures for such attacks must necessarily be carried out at the physical level as well---once these signals have been sampled and digitized, no amount of post-processing can repair the compromised sensor data.

Broadly speaking, sensors can be divided into two categories: passive (those that sense pre-existing physical signals) and active (those that perform some action to evoke and measure a physical response from some measurable entity).  Examples of passive sensors include temperature, humidity, and ambient light, while active sensors include ultrasound, laser scanners (LIDAR), and radar. 
Passive sensors are largely na\"ive listening devices--they will blindly relay information to higher levels of software without regard for the integrity of that information.  Digital filtering and other post-processing techniques can be used to remove noise from passive sensors, but they remain unable to combat attacks at the physical layer in any meaningful way.  On the other hand, active sensors introduce the possibility for more advanced security measures. PyCRA is, at its core, a method of ensuring the trustworthiness of information obtained by active sensors by comparing their responses to a series of physical queries or challenges. The driving concept behind PyCRA is that, by stimulating  the environment with a randomized signal and measuring the response, we can ensure that the signal measured by the sensor is in accordance with the underlying sensing physics. The randomization in the  stimulating signal is known to the active sensor but unknown to the adversary \footnote{Note that the randomness is purely private and there is no exchange/communication of it is needed.}.
This randomized stimulation and subsequent behavioral analysis---the physical challenge-response authentication---is the main contribution of this work.  

We further extend the resilience of PyCRA against passive attacks by means of a \emph{confusion phase}---a period of low signal to noise ratio in which the ability to detect and respond to a physical challenge is made more difficult for any attacker. This additional phase leverages theoretical guarantees from the literature on point change detection, in which one party (the attacker, in our case) attempts to detect the point at which a signal randomly changes amplitudes in the presence of noise. An intelligent attacker could, upon sensing a physical challenge, attempt to respond in a timely manner by spoofing the measured, physical signal.  By increasing the time required to detect each challenge, the confusion phase provides theoretical guarantees for PyCRA's resilience to active sensor attacks. 

We demonstrate the effectiveness of PyCRA for three exemplary cases: physical attack detection for magnetic encoders, physical attack resilience for magnetic encoders, and passive eavesdropping detection for RFID readers. Magnetic encoders are used in a wide array of commercial and industrial applications and are representative of a large class of inductive active sensors. We demonstrate not only how active spoofing attacks can be detected for these inductive sensors but also how the effects of these attacks can be mitigated. Eavesdropping detection on RFID readers serves to illustrate an extension of PyCRA to enable detection of \emph{passive} attacks. We believe that the methods demonstrated in this work can be applied to a broad array of active sensors beyond those studied directly in this work, including ultrasound, optical sensors, active radar, and more.

\begin{figure*}
\begin{center}
	\includegraphics[width=0.65\textwidth]{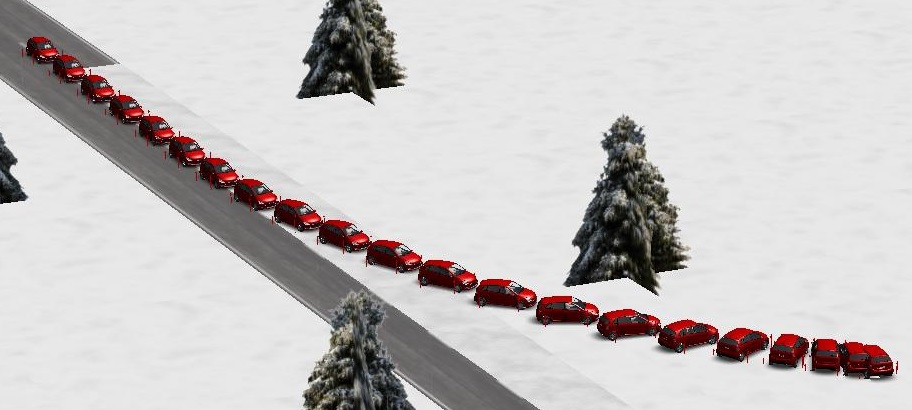}
\end{center}
\caption{The consequence of applying an ABS sensor spoofing attack while braking over ice. This simulation shows the position of the attacked car over multiple time instances.}
\label{fig:carsim}
\vspace{-3mm}
\end{figure*}

\subsection{A Motivating Example}
In order to illustrate the potential consequences of physical spoofing attacks on individual sensors, we will look at the example of the Antilock Braking System found in modern cars. The sensors used in these systems are responsible for measuring the angular rate of rotation for individual wheels in order to detect when a wheel is skidding.  When a car begins to skid, the ABS prevents the wheels from locking up and thus gives the driver improved control over the vehicle in otherwise dangerous situations.  If one or more of these sensors is compromised, the ABS can be tricked into thinking that the car is skidding when it is not or that it is operating nominally when in fact the car has entered a potentially life-threatening skid.  Using a commercial vehicle simulation program called CarSim \cite{carsim}, we can visualize the effects of compromising the ABS system of a car. Figure \ref{fig:carsim} shows the result of an attack in a series of time-lapsed images.  In this simulation, the car is proceeding in a straight line when it encounters a patch of ice and begins to apply the brakes.  The vehicle begins to skid at some point, but the driver retains control thanks to the ABS.  Shortly after the brakes are applied, the right, rear ABS sensor is spoofed such that it reports a slower speed, allowing the car to skid uncontrollably off the road. Work reported in \cite{YasserABS} has shown a physical implementation of such an attack.

In this paper we use ABS spoofing as a motivating example for the development of PyCRA, but it is important to note that securing ABS sensors with physical challenge-response authentication is merely a specific application of a broader security scheme---periodic injections of known stimuli into a physical system allow a sensor to monitor the trustworthiness of the perceived environment  by comparing observed behavior to ideal or predicted \emph{responses}.

\subsection{Contributions of PyCRA }
In summary, the contributions described in this paper are multi-fold:\vspace{-2mm}
\begin{itemize}
\item We present a generalizable physical challenge-response authentication scheme for active sensing subsystems.
\item We extend the basic concept of  physical challenge-response authentication for detecting the presence of passive attacks and providing resilience against active physical attacks.
\item We extend PyCRA with a novel security mechanism known as the \emph{confusion phase}, which limits the attacker's capability to counter-measure the physical challenge-response authentication scheme.
\item We provide rigorous mathematical results describing how the \emph{confusion phase} enhances the performance of PyCRA while adding fundamental limitations to the capabilities of active physical attacks.
\item We demonstrate the effectiveness of PyCRA, our implementation of physical challenge-response authentication, against several different attack types with three exemplary applications: (1) detection of active attacks on magnetic encoders, (2) resilience against active attacks on  magnetic encoders, and (3) detecting passive eavesdropping attacks on RFID readers.
\end{itemize}
The rest of this paper is organized as follows. Section \ref{sec:attacker_model} outlines the  attacker model. Section \ref{sec:pycra} describes basic operation of the PyCRA authentication scheme for detecting active attacks. Sections \ref{sec:ext_resilience} and~\ref{sec:sniffing} extend PyCRA to other applications, namely providing resilience against active attacks and detecting passive attacks. In order to enhance more the performance of PyCRA and increase its security, we introduce a novel design mechanism to PyCRA named the \emph{confusion phase}. Details of this novel mechanism along with theoretical analysis of security guarantees provided by this mechanism is detailed in Section~\ref{sec:theory} and Appendix~\ref{sec:proofodDetDelayDecay}.
Sections \ref{sec:absattacks}, \ref{sec:absattacks2} and \ref{sec:resultsRFID} are devoted to the results of three case studies: Section \ref{sec:absattacks} discusses attack detection for magnetic encoders; Section \ref{sec:absattacks2} describes how PyCRA authentication provides resilience against physical attacks on magnetic encoders; and Section \ref{sec:resultsRFID} shows the results of extending PyCRA to the detection of passive eavesdropping attacks on RFID readers.  Finally, we offer a discussion and  concluding thoughts in Sections \ref{sec:discussion} and \ref{sec:conclusion}.

A preliminary version of this paper appeared in~\cite{ShoukryPyCRA}, providing an explanation of only the physical challenge authentication mechanism itself. In this paper, we discuss the details of extending PyCRA to counteract \emph{passive} physical attacks and to provide \emph{resilience} against active physical attacks (Sections~\ref{sec:ext_resilience} and~\ref{sec:sniffing}) along with more experimental results demonstrating the performance of these extensions (Section~\ref{sec:absattacks2}). Finally, this paper introduces in detail the notion of a~\emph{confusion phase}, as briefly mentioned in~\cite{ShoukryPyCRA} along with the underlying mathematics shown in Section~\ref{sec:theory} and Appendix~\ref{sec:proofodDetDelayDecay}.


\section{Attacker Model}
\label{sec:attacker_model}

\begin{figure}
\centering
{
	\includegraphics[width=0.75\columnwidth]{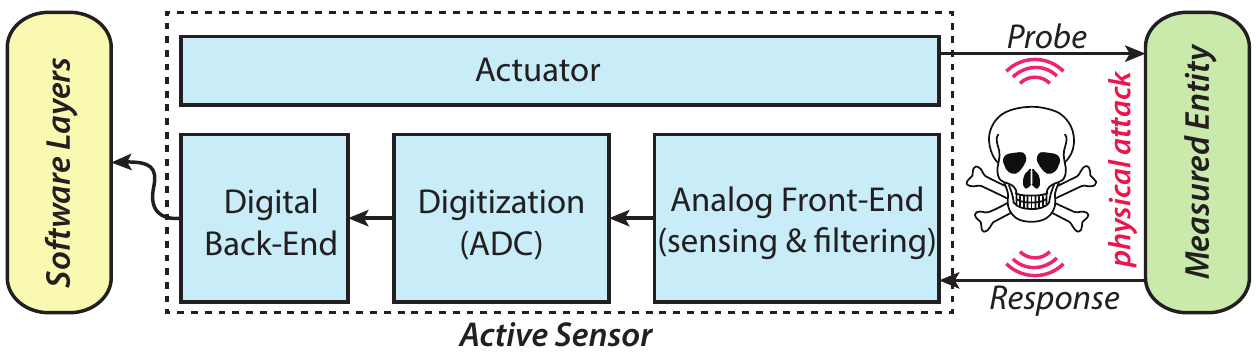}
}
\caption{\label{fig:sensor}
A typical active sensor architecture. The actuator generates an analog signal (energy) which is reflected by the measured entity back to the sensor. The received analog signal is captured and processed by the analog front-end. The signal is then converted to a digital format which is processed once more (by the digital back-end) before being sent to higher level software layers.}
\vspace{-4mm}
\end{figure}

Before describing mechanisms by which we can detect and prevent sensor attacks at the physical layer, we must differentiate between two broad categories of sensors---namely passive and active sensors---and define what we mean by a physical attack.

\subsection{Passive vs. Active Sensors}

Sensors can be broadly classified as either passive or active based on the source of energy being sensed. Passive sensors measure ambient energy. For example, temperature sensors like those found in thermostats are considered passive, because they measure heat energy in the ambient environment. By contrast, active sensors probe some physical entity with self-generated energy as shown in Figure \ref{fig:sensor}. This energy is partially reflected back to the sensor where it is measured and used to infer properties about some physical phenomenon. Examples of active sensors include ultrasonic range finders (used in robotics), optical and magnetic encoders (used in automotive vehicles, industrial plants, \& chemical refineries), radar, and even radio-frequency identification (RFID) systems. In RFID, a reader is used to generate electromagnetic waves which are then used by wireless tags to transfer back their unique identifier to the reader.

In this paper, we focus on providing security for active sensors. In particular, we leverage an active sensor's ability to emit energy in order to 1) provide detection of active attackers trying to spoof the sensor, 2) mitigate the effects of active spoofing attacks and 3) detect passive eavesdropping attacks attempting to listen to the information received by the sensor. In the following subsections, we define what we mean by physical attacks on active sensors and outline the assumed properties and limitations of a potential adversary.

\subsection{Defining Physical Attacks}

In this paper, a physical attack refers to a malicious alteration of a physical, analog signal (e.g., magnetic waves, acoustic waves, visible waves) prior to transduction and digitization by a sensor, as shown in Figure \ref{fig:sensor}.

\subsection{Adversarial Goals}
The adversary considered in this work has a number of goals related to misinforming and misleading sensors.  These goals are summarized below.

\begin{enumerate}
\item[\textbf{G1}] \emph{\underline{Concealment}: An attacker does not want the presence of his or her attack to be known.} 
\end{enumerate}
\noindent If a sensor attack can be easily detected, preventative countermeasures like hardware redundancy and resilience at the system-level can often be used to mitigate the damage done by the attack \cite{Yasser_SMT,Bullo_TAC}.  

\begin{enumerate}
\item[\textbf{G2}] \emph{\underline{Signal Injection}: An attacker will attempt to trick the sensor into thinking that a malicious, injected signal is the true physical signal. }
\end{enumerate}
\noindent The primary goal of an attack is to replace the true physical signal that a sensor aims to sense with a malicious signal.  In other words, an adversary will attempt to ``inject'' a signal into the physical medium that the sensor is measuring in order to jam or spoof the sensor.
\begin{enumerate}
\item[\textbf{G3}] \emph{\underline{Signal Masking}: An attacker will attempt to prevent the sensor from being able to detect the true physical signal.} 
\end{enumerate}
\noindent If the sensor is still capable of reliably discerning the correct signal from the malicious, injected signal, then the attack may not be successful. Thus, the adversary aims not only to inject a signal but also to mask the true signal, whether by overpowering, modifying, or negating (canceling) it.

\begin{figure}
\centering
{
	\includegraphics[width=0.4\columnwidth]{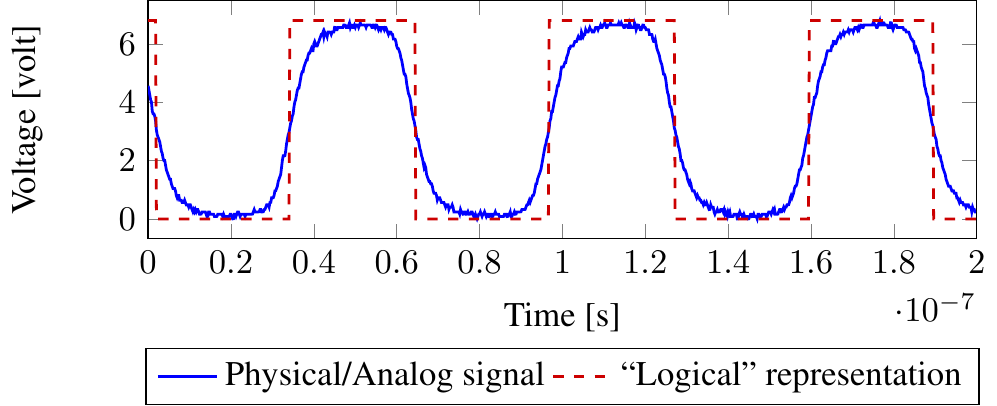}
	\includegraphics[width=0.4\columnwidth]{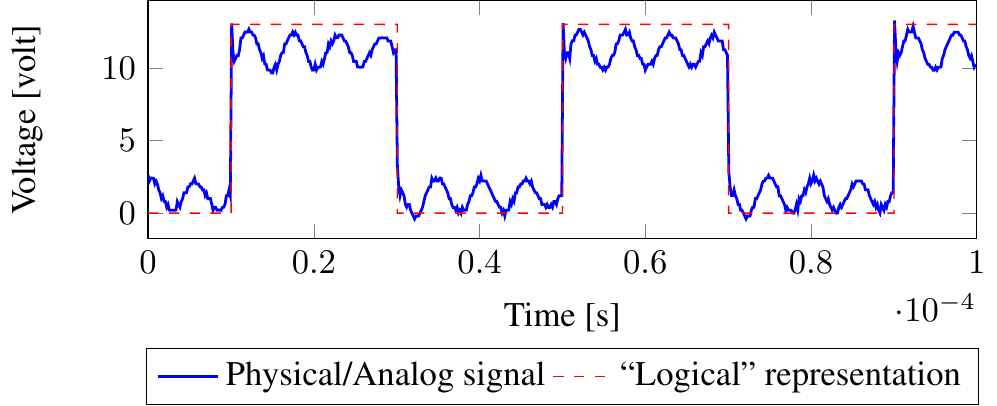}
	}
	\caption{\label{fig:delays}Examples of physical delays seen in typical sensing and actuation hardware, including optical sensors (left) and  electromagnetic coupled (e.g., RFID) sensors (right). In each case, the measured analog signal (blue solid) lags behind the ideal, ``logical'' signal (red dashed), causing delays. }
\vspace{-4mm}
\end{figure}

\subsection{Assumptions about the Adversary}

The physical attacks against sensors considered in this work operate under four main assumptions:
\begin{enumerate}
\item[\textbf{A1}]\emph{\underline{Non-invasiveness}: Attacks are of a \emph{non-invasive} nature---that is, the attacker is not allowed direct access to the sensor hardware.  Additionally, the adversary does not have access to the sensor firmware or software, whether directly or through wired or wireless networking.} 
\end{enumerate}
\noindent In most life- and safety-critical applications, engineers are careful to ensure that sensors are not physically exposed and vulnerable to direct tampering. For example:
\begin{itemize}
\item Sensors are often installed inside the body of a physically secured infrastructure (e.g., sensors inside the body of an automotive system, moving UAV drones, etc.).
\item For sensors which are physically accessible, existing techniques in the literature demonstrate ways to implement tamper-proof packaging to protect sensors from direct, physical modifications \cite{1260985,tamperProofSmartCards,AndersonTamperProof}.
\item Numerous sensor systems have methods for detecting when wires connecting their various sensors have been tampered with. For example, automotive systems are equipped with sensor failure detection systems which can detect whether all sensor subsystems are correctly connected and alert the driver if any of them fails 
\cite{klassen1993fault}.
\end{itemize}

\noindent Because of this, any attack must be carried out from a distance, without direct access to any sensor hardware.  In short, an adversary is assumed to have access only to the physical/analog medium used by the sensor---magnetic waves, optics, acoustics, etc. 

Additionally, it is important to distinguish these sensors from \emph{sensor nodes} (which appear in the literature of sensor networks); the attacks and countermeasures in this work target \emph{sensors} themselves. Sensors are simple subsystems designed to perform only one simple task; sensing the physical world. Because of this, many sensors do not support remote firmware updates 
and do not typically receive commands from a remote operator, making such attack vectors uncommon as many sensors do not have such capabilities.

\begin{figure*}
\centering
{
	\begin{tabular}{c|c|c}
	
	\subfloat[]{\label{fig:eavesdrop_attack}
	\includegraphics[width=0.30\textwidth]{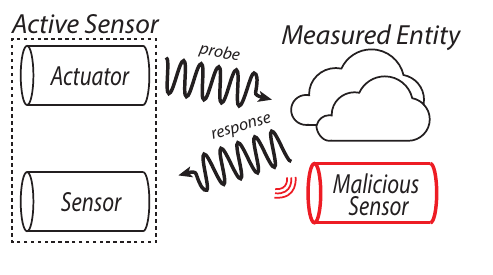}}&
	
	\subfloat[]{\label{fig:simplistic_attack}
	\includegraphics[width=0.30\textwidth]{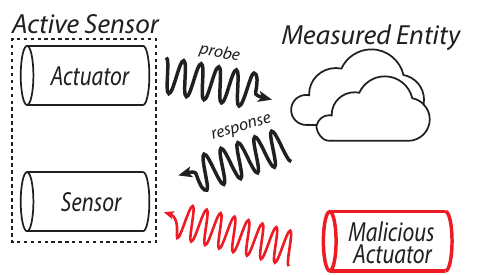}}&
	
	\subfloat[]{\label{fig:advanced_attack}
	\includegraphics[width=0.30\textwidth]{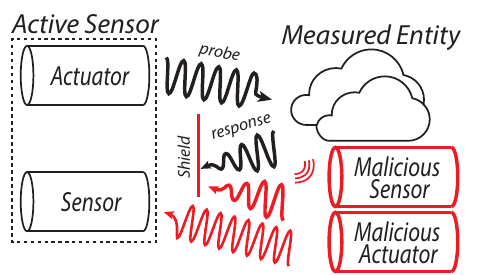}}
	
	\end{tabular}
	}
	\caption{An illustration of three physical attack types: (a) a passive eavesdropping attack, (b) a simple spoofing attack where a malicious actuator blindly injects a disruptive signal, and (c) an advanced spoofing attack where an adversary uses a sensor to measure the original signal and an actuator to actively cancel the original signal and inject a malicious one.  \label{fig:attack_types}}
\vspace{-4mm}
\end{figure*}

\begin{enumerate}
\item[\textbf{A2}]\emph{\underline{Trusted Measured Entity}
We assume that the physical entity to be measured by the sensor is trusted and incapable of being compromised.}
\end{enumerate}
\noindent Similar to the sensor hardware itself, the entity that the sensor aims to measure is typically difficult to access or alter directly while maintaining Goals G1--G3.  For example, in RFID systems the tag itself is often encased in tamper-proof packaging \cite{1260985,tamperProofSmartCards}; for ultrasonic ranging and active radar, maliciously altering the measured entity (often the entire surrounding environment) is impractical in time \& effort and undoubtedly violates Goal G1; for airplane engine speed sensors, the engines cannot easily be modified or replaced; for heart monitors, the heart cannot (we hope) be modified \cite{ghosttalk}, and so forth.

\begin{enumerate}
\item[\textbf{A3}]\emph{\underline{Physical Delays ($\tau_{attack}$)}: Adversaries require physical hardware with inherent physical delays.
This delay, though variable in duration, is fundamental to all physical actuation and sensing hardware. }
\end{enumerate}

\noindent These same analog/physical signals cannot be manipulated or even observed (i.e. sniffed) without physical hardware. That is, to tamper with magnetic waves, an attacker needs hardware that is able to generate magnetic waves, optical signals need physical hardware that generates optical signals, and so on. Furthermore, this hardware has to obey fundamental physics imposed by nature; the underlying physics dictate that the response of any physical element is governed by a dynamical model (mathematically modeled using differential/difference equations) \cite[ch. 2]{FradenBook}, \cite[chs. 8--9]{brauer2006magnetic}. This dynamical model describes the output response for each physical element in response to their inputs, e.g., the time for a voltage to drop from a certain value to zero and so on.
Although from a system point of view, we often assume that analog signals like those in Figure \ref{fig:delays} take on logical values of 0 and 1, the underlying physics is always different from this ``system'' point of view. For example, Figure \ref{fig:delays} shows how hardware that generates clock waveforms and optical pulse signals behaves quite differently from the desired, logical signals used to control them.  In general, no physical signal can arbitrarily jump from one state to another without suffering from \emph{delays} imposed by physics \cite[ch. 2]{FradenBook}.

Furthermore, these physical delays are lower bounded by a non-zero, fundamental limit. For example, the time response of an electromagnetic sensor/actuator is a multiple of physical constants like magnetic permeability \cite[chs. 8--9]{brauer2006magnetic} or permitivity and electric constants for capacitive sensors \cite[ch. 4]{FradenBook}. In general, the time response of any sensor or actuator can never be below certain fundamental thresholds controlled by physical constants. We refer to this physical delay as $\tau_{attack}$ for the remainder of this paper. 

\begin{enumerate}
\item[\textbf{A4}]\emph{\underline{Computational Delays}: 
PyCRA is designed and analyzed with a focus on \emph{physical} delays. 
We make no assumption regarding the computational power of a potential adversary. }
\end{enumerate}
\noindent We assume that an adversary has knowledge of the underlying security mechanism, attempting to conceal an attack by reacting to each physical challenge or probe from the PyCRA-secured active sensor.  In practice, such an adversary would suffer from \emph{computational delays} in addition to the physical delays addressed above.  These delays would make it even more difficult for an adversary to respond to these challenges in a timely manner.  PyCRA is designed to leverage only the physical delays addressed above, but additional computational delays would make it even easier to detect the presence of an attack.

\subsection{Physical Attack Types for Sensors}
\label{sec:attack_types}
Attacks can be classified as either passive (eavesdropping) or active (spoofing). While we consider only  physical/analog attacks in accordance with assumptions A1--A4, the passivity of an attack is decided by whether or not the attacker is manipulating (or spoofing) the physical signal or merely listening to it. Active attacks themselves can be classified once more into simple spoofing or advanced spoofing attacks. In short,
physical sensor attacks in accordance with assumptions A1--A4 can be broadly divided into three categories (\underline{T}ypes):  
\begin{enumerate}
\item[\textbf{T1}]\emph{\underline{Eavesdropping Attacks}: 
\label{sec:eavesdropping}
In an eavesdropping attack, an adversary uses a malicious sensor in order to listen to the active sensor's ``communication'' with the measured entity (Figure \ref{fig:eavesdrop_attack})}.
\item[\textbf{T2}]\emph{\underline{Simple Spoofing Attacks}:
\label{sec:naivespoofing}
In a simple spoofing attack, an adversary uses a malicious actuator to blindly inject a malicious signal in order to alter the signal observed by the sensor.  These attacks are simple in that the malicious signal is not a function of the original, true signal (Figure \ref{fig:simplistic_attack}). }

\item[\textbf{T3}]\emph{\underline{Advanced Spoofing Attacks}
\label{sec:intelligentspoofing}
In an advanced spoofing attack, an adversary uses a sensor in order to gain full knowledge of the original signal and then uses a malicious actuator to inject a malicious signal accordingly. 
This enables an attacker to suppress the original signal or otherwise alter it in addition to injecting a malicious signal (Figure \ref{fig:advanced_attack}). }
\end{enumerate}
We argue that these attack types span all possible modes of attacks that abide by Assumptions A1--A4 with those goals outlined in G1--G3. For example, jamming or Denial of service (DoS) attacks falls in category T2 where the attacker's actuator is used to blindly generate high amplitude, wide bandwidth signals to interfere with the physical signal before it reaches the sensors; replay attacks fall in either category T2 or T3 based on whether the attacker is blindly replaying a physical signal or destructing the original physical signal before inserting the replay signal; spoofing attacks like those demonstrated in \cite{ghosttalk} fall in category T2; and attacks described in \cite{YasserABS} fall within both T2 and T3.

At first glance, attacks of type T1 may not seem important especially if the sensor under attack measures a physical signal that is publicly accessible (e.g., room temperature, car speed, etc.). In such cases, an adversary can measure the same physical signal without the need to ``listen'' to the interaction between the active sensor and the environment. However, this may not always be the case.  For example, an attacker might measure magnetic waves during an exchange between an RFID reader and an RFID tag, learning potentially sensitive information about the tag. These attacks are passive, meaning that the attacker does not inject any energy into the system. Sections~\ref{sec:absattacks} describes methods for detecting attacks T2 and T3, leaving attack type T1 for later discussion in Section~\ref{sec:resultsRFID}.


\section{The PyCRA Authentication Scheme}
\label{sec:pycra}
The core concept behind PyCRA is that of physical challenge-response authentication.  In traditional challenge-response authentication schemes, one party requires another party to prove their trustworthiness by correctly answering a question or \emph{challenge}.  This challenge-response pair could be a simple password query, a random challenge to a known hash function, or other similar mechanisms. In the proposed physical challenge-response authentication, the challenge comes in the form of a \emph{physical} stimulus placed on the environment by an active sensor.  Unlike traditional schemes, the proposed \emph{physical} challenge operates in the analog domain and is designed so that an adversary cannot issue the correct response because of immutable physical constraints rather than computational or combinatorial challenges.  

We begin by modeling the problem of detecting physical sensor attacks as an authentication problem. To draw this analogy, let us consider the communication system shown in Figure \ref{fig:system_diagram_normal}. This figure shows two `parties': (1) an active sensor composed of actuation and sensing subsystems and (2) the measured entity which responds to signals emitted by the actuator contained within the active sensor.  The first party---the active sensor---is responsible for initiating the ``communication'' by generating some physical signal such as a magnetic, acoustic, or optical wave. The second party---the measured entity---responds to this ``communication'' by modulating this signal and reflecting it back to the sensing subsystem of the active sensor. With this analogy in mind, the problem of detecting physical attacks can be posed as that of ensuring that the ``message'' seen by the sensor has originated from a trusted party (the true entity to be measured). This is akin to identity authentication in the  the literature of computer security but applied to the analog domain.

\subsection{Simple PyCRA Attack Detector}
Using the communication analogy shown in Figure \ref{fig:system_diagram_normal} and recalling that we are interested only in active sensors as described in Section \ref{sec:attacker_model}.1, we notice that the measured entity, as a participating party in this communication, is strictly \emph{{passive}}, i.e. it cannot initiate communication; it responds only when the sensor generates an appropriate physical signal.

\begin{figure*}[!t]
\centering
{
	\begin{tabular}{c|c|c}
	\subfloat[ ]{\label{fig:system_diagram_normal}
	\resizebox{0.3\textwidth}{!}{
	\begin{tikzpicture}
	
	\draw[very thick,->] (-1.8,-0.2) -- (-1.1,-0.2); 
	\node at (-1.5,0) {$u(t)$};
	\draw[very thick,->] (-1.1,-1.6) -- (-1.8,-1.6); 
	\node at (-1.5,-1.3) {$y(t)$};
	
	\node at (-0.3,-0.8) {\includegraphics[width=.1\textwidth]{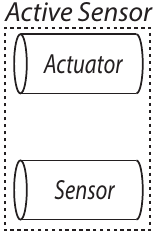}};
	
	\node at (3.8,-0.5) {\includegraphics[width=.1\textwidth]{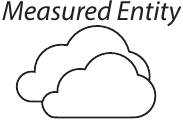}};

	\draw[very thick,->] (0.7,0) -- (2.5,-0.7); 
	\draw[->] (1,0) -- (2.5,0) node[anchor=north]{\small $t$};
	\draw[->] (1.2,-0.2) -- (1.2,0.7) node[anchor=west]{\small $\mathcal{A}(t)$}; 
	\draw[-] (1.2, 0.45) -- (2.45, 0.45); 

	\draw[very thick,->] (2.5,-1.1) -- (0.7,-1.8);
	\draw[->](2.5,-2.2) -- (1,-2.2) node[anchor=north]{\small $t$}; 
	\draw \sinewave{1.1}{0.085}{-2.2}{0.25};

	\draw[->](4.5,-2.2) -- (3,-2.2) node[anchor=north]{\small $t$}; 
	\draw[->] (4.3,-2.4) -- (4.3,-1.7) node[anchor=east]{\small $a(t)$}; 

	\end{tikzpicture}
	}
	}&
	\subfloat[ ]{\label{fig:system_diagram_noattack}
	\resizebox{0.3\textwidth}{!}{
	\begin{tikzpicture}
	
	\draw[very thick,->] (-1.8,-0.2) -- (-1.1,-0.2); 
	\node at (-1.5,0) {$u(t)$};
	\draw[very thick,->] (-1.1,-1.6) -- (-1.8,-1.6); 
	\node at (-1.5,-1.3) {$y(t)$};
	
	\node at (-0.3,-0.8) {\includegraphics[width=.1\textwidth]{activesensor}};
	
	\node at (3.8,-0.5) {\includegraphics[width=.1\textwidth]{environment}};

	\draw[very thick,->] (0.7,0) -- (2.5,-0.7); 
	\draw[->] (1,0) -- (2.5,0) node[anchor=north]{\scriptsize t};
	\draw[->] (1.2,-0.2) -- (1.2,0.7) node[anchor=west]{\small $\mathcal{B}(t)$}; 
	\draw[-] (1.2, 0.45) -- (1.7, 0.45); 
	\draw[-] (1.7, 0.45) -- (1.7, 0); 
	\draw[ultra thick,-,red!80!black] (1.7,0) -- (1.9,0);
	\draw[-] (1.9, 0.45) -- (1.9, 0); 
	\draw[-] (1.9, 0.45) -- (2.45, 0.45); 

	\draw[very thick,->] (2.5,-1.1) -- (0.7,-1.8);
	\draw[->](2.5,-2.2) -- (1,-2.2) node[anchor=north]{\small $t$}; 
	\draw \halfsinewave{1.8}{0.08}{-2.2}{0.25};
	\draw \halfsinewave{1.1}{0.08}{-2.2}{0.25};
	\draw[ultra thick,-,red!80!black] (1.59,-2.2) -- (1.79,-2.2);
		
	\draw[->](4.5,-2.2) -- (3,-2.2) node[anchor=north]{\small $t$}; 
	\draw[->] (4.3,-2.4) -- (4.3,-1.7) node[anchor=east]{\small $a(t)$}; 
	
	\end{tikzpicture}
	}
	}&
	\subfloat[ ]{\label{fig:system_diagram_attack}
	\resizebox{0.3\textwidth}{!}{
	\begin{tikzpicture}
	
	\draw[very thick,->] (-1.8,-0.2) -- (-1.1,-0.2); 
	\node at (-1.5,0) {$u(t)$};
	\draw[very thick,->] (-1.1,-1.6) -- (-1.8,-1.6); 
	\node at (-1.5,-1.3) {$y(t)$};
	
	\node at (-0.3,-0.8) {\includegraphics[width=.1\textwidth]{activesensor}};
	
	\node at (3.8,-0.5) {\includegraphics[width=.1\textwidth]{environment}};

	\draw[very thick,->] (0.7,0) -- (2.5,-0.7); 
	\draw[->] (1,0) -- (2.5,0) node[anchor=north]{\small $t$};
	\draw[->] (1.2,-0.2) -- (1.2,0.7) node[anchor=west]{\small $\mathcal{B}(t)$}; 
	\draw[-] (1.2, 0.45) -- (1.7, 0.45); 
	\draw[-] (1.7, 0.45) -- (1.7, 0); 
	\draw[ultra thick,-,red!80!black] (1.7,0) -- (1.9,0);
	\draw[-] (1.9, 0.45) -- (1.9, 0); 
	\draw[-] (1.9, 0.45) -- (2.45, 0.45); 

	\draw[very thick,->] (2.5,-1.1) -- (0.7,-1.8);
	\draw[->](2.5,-2.2) -- (1,-2.2) node[anchor=north]{\small $t$}; 
	\draw \attackwave{1.8}{0.027}{-2.2}{0.15}{0.03};
	\draw \attackwave{1.1}{0.027}{-2.2}{0.15}{0.03};
	\draw[ultra thick,-,red!80!black] (1.59,-2.2) sin (1.65,-2.30) cos (1.70,-2.2) sin (1.75,-2.10) cos (1.80,-2.2);   

	\draw[->](4.5,-2.2) -- (3,-2.2) node[anchor=north]{\small $t$}; 
	\draw[->] (4.3,-2.4) -- (4.3,-1.7) node[anchor=east]{\small $a(t)$}; 
	\draw \sinewave{3.2}{0.04}{-2.2}{0.1};
	\draw \sinewave{3.68}{0.04}{-2.2}{0.1};
	
	\end{tikzpicture}
	}
	}
	\end{tabular}
}

\caption{\label{fig:system_diagram}An illustration of the PyCRA architecture and attack detection scheme: (a) During normal operation, the active sensor generates a signal $\mathcal{A}(t)$. This signal passes through environmental dynamics and is reflected back to the sensor as $y(t)$; (b) Using the proposed PyCRA scheme, the sensor generates a modulated signal $\mathcal{B}(t)$.  If there is no attack present, the reflected signal diminishes if the active sensor's actuator is driven to zero; (c) Using the proposed PyCRA scheme while the sensor is under attack (by signal $a(t)$), a malicious signal is detected during the period when the actuator is disabled.} 
\vspace{-3mm}
\end{figure*}

PyCRA exploits this ``passivity'' in order to facilitate the  detection of attacks.  Without PyCRA, an active sensor's actuator would probe the measured entity in a normal fashion using a deterministic signal denoted by $\mathcal{A}(t)$.  We embed in this signal a physical challenge through pseudo-random binary modulation of the form:
\begin{equation}
\mathcal{B}(t) =  u(t)\mathcal{A}(t), ~~~u(t) \in \{0,1\}
\label{eq:probe}
\end{equation}
\noindent where $u(t)$ is the binary modulation term and $\mathcal{B}(t)$ is the modulated output of the actuator.  The output of the active sensor is denoted by $y(t)$ as shown in Figure \ref{fig:system_diagram}.  In the absence of an attacker and from the passivity of the measured entity, setting $u(t) = 0$ (and consequently $\mathcal{B}(t) = 0$) at time $t_{challenge}$ will cause $y(t)$ to go to zero.  

Potential attackers must actively emit a signal $a(t)$ to overpower or mask $y(t)$ (Goals G2--G3).  A na\"ive attacker might continue to emit this signal even when $\mathcal{B}(t) = 0$ as shown in Figure \ref{fig:system_diagram_attack}.  In this case, the attack can be easily detected, since any nonzero $y(t)$ while $u(t) = 0$ can be attributed to the existence of an attacker. 

More advanced attackers might attempt to conceal their attacks when they sense the absence of $\mathcal{B}(t)$ as in Goal G1. Due to Assumption A3, an attacker could drive $a(t)$ to zero only after a delay of $\tau_{attack}$, where $\tau_{attack} \ge \tau_{physical \; limit} > 0$ is the unavoidable physical delay inherent in the attacker's hardware. Therefore, the mechanism described above can still detect the presence of an attack within this unavoidable time delay.  Furthermore, an attacker cannot learn and compensate for this inherent delay preemptively due to the randomness of the modulation term $u(t)$. Again, any nonzero $y(t)$ sensed while $u(t) = 0$ can be attributed to the existence of an attacker. The simple PyCRA attack detector can be summarized as follows:\\
{\textbf{[Step 1]} Select a random time, $t_{challenge}$}\\
{\textbf{[Step 2]} Issue a physical challenge by setting $u(t_{challenge}) = 0$ }\\
{\textbf{[Step 3]} If $y(t_{challenge}) > 0$, declare an attack}\\
Note that the previous process needs to happen within small amount of time (e.g., in the order of milliseconds) such that it does not affect the normal operation of the system.


\subsection{$\chi^2$ PyCRA Attack Detector}

As with the attacker, the actuator used by the active sensor itself suffers from physical delays. This means that when PyCRA issues a physical challenge, the actuator output does not transition immediately. Apparently, if the physical delay in the active sensor is greater than $\tau_{attack}$, then an adversary can conceal his signal.  To counter this, PyCRA constructs a mathematical model for the sensor that is used---in real time---to predict and eliminate the effects of the active sensor's physics. By calculating the residual between the expected output and the measured output, PyCRA can still detect the existence of an attack. The details of this procedure along with an experimental example are the subject of this subsection.

\subsubsection{Obtaining the Sensor Model}
To compensate for the actuator dynamics, we first need to acquire an accurate model that captures the underlying physics of the active sensor. Below we model the active sensor using the generic nonlinear state update of the form:
\begin{align}
	x(t+1) &= f(x(t), u(t)) + w(t) \label{eq:sys1}\\
	y(t) &= h(x(t)) + v(t) \label{eq:sys2}
\end{align}
\noindent where $x(t) \in \R^{n}$ is the active sensor state at time $t \in \N_0$ (e.g., the electrical current and voltages inside the sensor at time $t$), $u(t) \in \R$ is the modulation input to the sensor, the function \mbox{$f:\R^{n} \times \R \rightarrow \R^{n}$} is a model describing how the physical quantities of the sensor evolve over time, and the function $h:\R^{n} \rightarrow \R$ models the sensor measurement physics. Such models can be either derived from first principles \cite{FradenBook,brauer2006magnetic,grimes2006encyclopedia} or through experimental studies \cite{ljung1998system,Landaue_Book}. 
Additionally, these models are used to design the sensors themselves and are typically known to the sensor manufacturers.
Finally, since no mathematical model can capture the true system behavior exactly, the term $w(t) \in \R^{n}$ represents the mismatch between the true sensor and the mathematical model while $v(t)$ models the noise in the sensor measurements.

\subsubsection{$\chi^2$ Detector}
\label{sec:sub:ch2detector}
We use the dynamical model of the sensor (Equations ~\eqref{eq:sys1} and~\eqref{eq:sys2}) in designing a $\chi^2$ detector to detect the existence of an attacker. $\chi^2$ detectors appear in the literature of automatic control, where they are used in designing fault tolerant systems \cite{MiroslavCDC,Mehra1971637,Willsky1976601}. The $\chi^2$ detector works as follows:\\
\textbf{[Step 1]} Select a random time, $t_{challenge}$.\\
\textbf{[Step 2]} Issue a physical challenge by entering the silent phase at time $t_{challenge}$.\\
\textbf{[Step 3] Residual Calculation:}
Here we use Equations ~\eqref{eq:sys1}~and~\eqref{eq:sys2} to calculate an estimate for the current sensor state $\hat{x}(t)$ and the predicted output $\hat{y}(t) = h(\hat{x}(t))$. This operation is initiated at $t_{challenge}$ when $u(t)$ transitions to 0---the actuator ``silence time''---and terminates once $u(t)$ transitions back to one, signaling the end of actuator ``silence.''

The model represented by Equations ~\eqref{eq:sys1} and \eqref{eq:sys2} describes the output of the sensor when the attack is equal to zero. Therefore, the residual{\footnote{The name of the Chi-squared ($\chi^2$) detector follows from the fact that, in the case of no attack, the residual $z(t)$ is a Gaussian random variable, and hence its square $g(t)$ is a $\chi^2$ distributed random variable.}} between the measured output and the predicted output, $z(t) = y(t) - \hat{y}(t)$, corresponds to both the attack signal as well as the environmental dynamics during the time interval before $u(t)$ drops to 0. 
For each segment of length $T$ where $u(t) = 0$, we calculate the norm of the residual $z(t)$ as:
\begin{equation} 
g(t) = \frac{1}{T}\sum_{\tau = t - T+1}^{t}  z^2(\tau)
\label{eq:residual}
\end{equation}
\textbf{[Step 4] Detection Alarm:} Once calculated, we compare the $\chi^2$ residual $g(t)$ against a pre-computed alarm threshold $\alpha$. This alarm threshold is chosen based on the noise $v(t)$. Whenever the condition $g(t) > \alpha$ is satisfied, the sensor declares that an attacker has been detected.

\begin{figure*}
\centering
{
	\begin{tabular}{c|c}
	\subfloat[Attack signal]{\label{fig:pycra_out}
	\includegraphics[width=0.41\columnwidth]{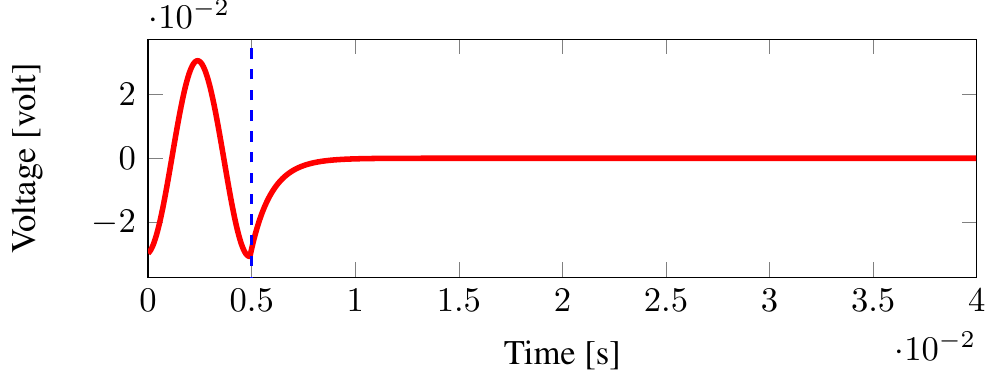}}&
	
	\subfloat[Sensor output]{\label{fig:pycra_out}
	\includegraphics[width=0.41\columnwidth]{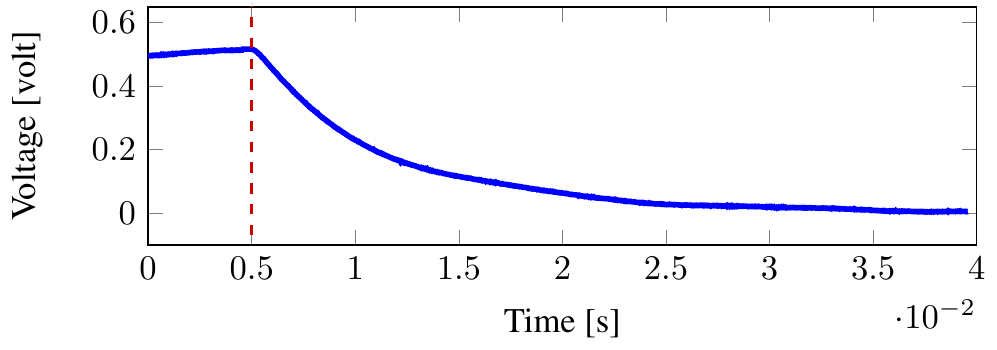}}\\\hline
	
	\subfloat[Residual = $\vert$ output - expected $\vert$]{\label{fig:pycra_residual}
	\includegraphics[width=0.41\columnwidth]{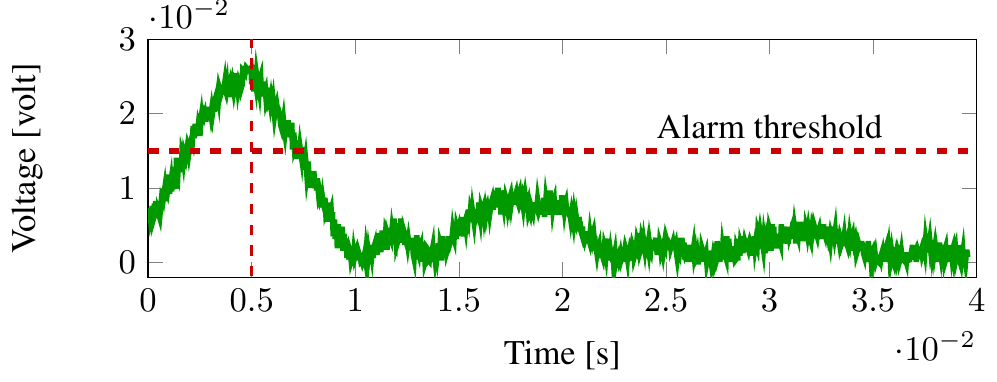}}&
	
	\subfloat[Expected sensor output per sensor model]{\label{fig:pycra_response}
	\includegraphics[width=0.41\columnwidth]{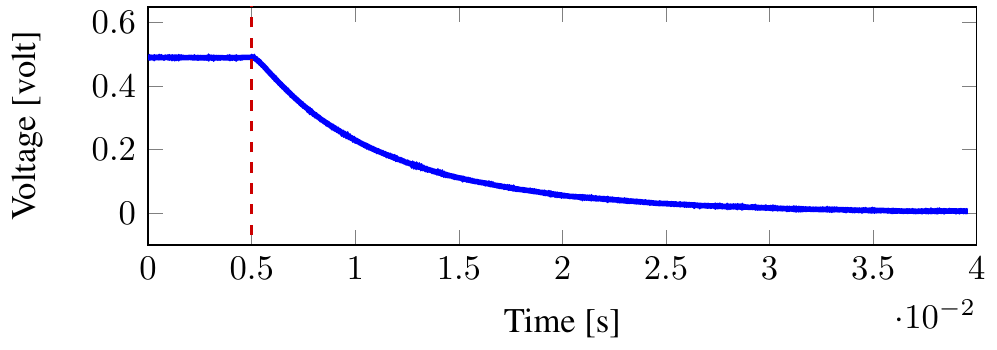}}

	\end{tabular}
	}
	\caption{\label{fig:pycra}Example showing attacker and sensor dynamics along with attack detection: (a) A smart attacker is spoofing the sensor. At time  $t = 5ms$, PyCRA issues a challenge by halting its sensing actuator ($u(t) = 0$). Accordingly, the smart attacker stops his attack signal as well. However, due to his physical delay, the attack signal takes 2 ms to reach zero. On the right, (b) shows that the sensor itself suffers from physical delays (15x slower than the attacker delay!), transitioning to zero after $35 ms$. However, due to knowledge of the sensor model shown in (d), we are able to remove the effect of the sensor's physical delay and discover that the sensor was under attack due to the high residual at time $t = 5$ ms (c). }
\vspace{-3mm}
\end{figure*}

\subsubsection{Experimental Example}

To better illustrate the operation of the proposed $\chi^2$ detector, consider the example shown in Figure \ref{fig:pycra}. In this example, an advanced attacker is spoofing a sensor. The attacker is aware of the existence of the PyCRA security scheme and tries to conceal the attacking signal. Recall that PyCRA issues physical challenges at random times (Equation \ref{eq:probe}). Therefore, the attacker is monitoring the output of the active sensor's actuator, $\mathcal{B}(t)$ (Figure \ref{fig:pycra} (b)) and once it detects that the sensor output is decaying at time = 5ms (as a consequence of switching the actuator from On to Off), the attacker immediately switches off his device. However, due to the physical delay in his device (Assumption A3), it takes him 2 ms for his signal to completely disappear (Figure \ref{fig:pycra} (a)). In this particular example, we designed the attacker's hardware to be 15x faster than the sensor itself. This can be seen by comparing the decay of the sensor signal  (Figure \ref{fig:pycra} (b)) with that of the attacker (Figure \ref{fig:pycra} (a)).

To overcome the slow dynamics of the sensor, PyCRA uses the sensor model (Equations \ref{eq:sys1} and \ref{eq:sys2}) to estimate the output of the sensor $\hat{y}$ after the challenge is issued. This estimation is shown in Figure \ref{fig:pycra} (d). At each point in time, PyCRA calculates the residual error (Equation \ref{eq:residual}) as shown in Figure \ref{fig:pycra}(c). It is apparent from  Figure \ref{fig:pycra}(c) that, because of the physical delays at the attacker, the residual exceeds the alarm threshold $\alpha$ when the physical challenge is issued at time = 5ms indicating the detection of an attack. The case studies shown in Sections \ref{sec:absattacks}-\ref{sec:resultsRFID} demonstrate in greater detail the effectiveness of this detection mechanism for several example applications and various attacks.

\section{Attack Resilience and Recovery using PyCRA} \label{sec:ext_resilience}

In the previous section, we discussed how to use physical challenge-response authentication to detect the existence of an attacker. While the proposed detection scheme can be used for a variety of sensor types, the same technique can be potentially extended to other applications as well.
In this section, we show how the PyCRA detection scheme can be applied to providing attack resilience against attack types T2 and T3. That is, our objective is not just to detect the existence of the attack, but also to recover (or estimate) the original sensor measurements.

Recall that PyCRA is based on the idea that by comparing the measured sensor signal with the predictive models described in Equations \eqref{eq:sys1} and \eqref{eq:sys2} we arrive at the residual $z(t)$.  This residual was used for detecting the existence of attacks. However, this same residual effectively provides an estimate of the attacker signal $a(t)$.  In the case of a jamming attack or similar ``dumb'' attacks, $a(t)$ may be a non-structured, noisy signal.  
However, as far as the spoofing attacks T2 and T3 (described in Section~\ref{sec:attack_types}) are concerned, the goal of the attacker is to strategically mask the true signal and/or inject a malicious signal (Goals G2-G3).  In these cases, $a(t)$ must be structured in a certain domain for short periods of time and it has to follow the same structure of the original sensor measurements.  For example, if the original sensor measurements consists of slowly varying sinusoids, then for a spoofing attack to be successful, the attack signal $a(t)$ must follow the same structure and needs to consist as well of slowly varying sinusoids (with a different frequencies in order to spoof the sensor). In this case, it may be possible to build a model for $a(t)$ from $z(t)$ in order to subtract its effects from the measured signal $y(t)$ and thus counteract the attack as we discuss in the next subsections.


\subsection{Simple PyCRA Resilient Estimator}
\label{sec:res_simple}
For sake of simplicity, we illustrate the concept of PyCRA resilient estimator to the case when the original sensor measurement and the attack as well is sinusoidal wave dominated by a single frequency component. Such structure appears in many magnetic and optical encoders used in many industrial and automotive applications to measure the rotational speed of moving objects.

Given the previous signal structure, if we consider the frequency domain representation of the measured signal over a window, we expect to see the energy of the signal concentrated at one frequency corresponding to the tone ring frequency. In other words, the magnitude of the frequency domain representation of the signal consists of only one ``peak''.

However, in the existence of an attacker, and using Fourier analysis, we can reasonably expect to observe energy concentrated at multiple frequencies and hence more than one peak (Figure \ref{fig:twopeaks}). Only one of those peaks corresponds to the frequency of the tone ring while all other energy corresponds to the attacker signal. Therefore, the sensor needs to be  able to distinguish between the correct tone ring frequency and the attacker frequency. 

We again model the resilience problem as an authentication problem, building on top of the PyCRA attack detection scheme. In this case, the multiple peaks in the frequency domain correspond to multiple parties claiming to be the tone ring. If the sensor is able to successfully authenticate the identity of these peaks, it will remain resilient to such attacks. 

Returning to the physical challenge-response authentication for resilience, we exploit once more the ``passivity'' property of the measured entity but now from the perspective of the frequency domain. It follows from the ``passivity'' of the tone ring that the energy at the frequency corresponding to the tone ring shall decay in correspondence to the physical model directly. Hence, for an attacker to be stealthy (Goal G1), he or she is obliged to control this energy such that it behaves in a manner consistent with the natural response of the tone ring.

To further explain the proposed mechanism, consider the example shown in Figure \ref{fig:timelapse}. Figure \ref{fig:timelapse} shows the frequency domain representation of the errors between the measured sensor values and the expected sensor values over a window of length $N$, immediately after a challenge is issued ($u(t)$ transitions to 0). The signal energy is concentrated at two frequencies. The first one at 50 Hz corresponds to the original signal while the other one at 120 Hz corresponds to the attack signal. The sensor does not know \emph{a priori} which frequency corresponds to the measured entity and therefore needs to authenticate the identity of these two frequencies. Since the measured entity is passive, the energy located at the frequency corresponding to the measured entity shall be zero during the silence time. Immediately after the sensor enters a silent period, the sensor starts to calculate the error between the expected natural response of the tone ring and the measured output. As shown in Figure \ref{fig:timelapse}, at the beginning of the silence time the error is equal to zero at both candidate frequencies. Even if the attacker can detect the silence time immediately and halt his attack, the dynamics of the actuator hardware used by the attacker will take some time to subside due to the un-mutable physical delay (Assumption A3). Within this time, the energy of the attack signal still exists in the measured signal (Assumption AR2). Therefore, as time continues, the error  in the energy starts to accumulate in the frequencies corresponding to the attacker signal.

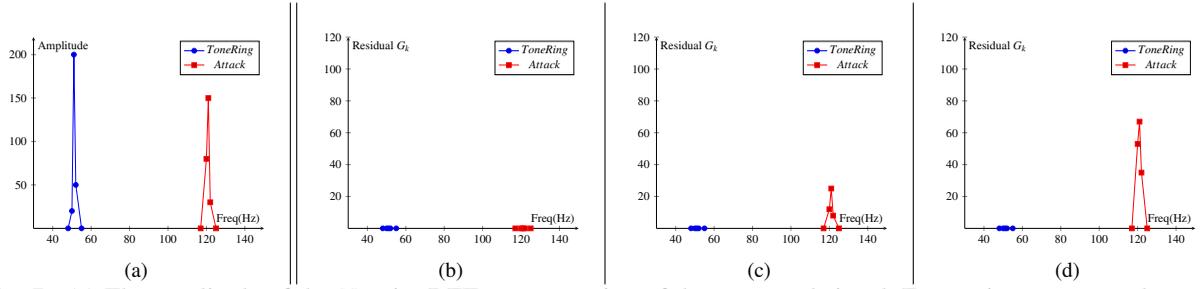
\begin{figure*}[!t]
\centering
{
	\begin{tabular}{c||c|c|c}
	\subfloat[]{\label{fig:twopeaks}
		\resizebox{0.21\textwidth}{!}{\begin{tikzpicture}
			\begin{axis}[xlabel={Freq(Hz)},ylabel={Amplitude},xmin=30, 	xmax=150, ymax = 220,axis lines={center},]
				\addplot coordinates {(48,0) (50,20)    (51,200)   (52,50) (55,0)};
				\addplot coordinates{(117,0)	(120,80)    (121,150)    (122,30) (125,0)};
				\legend{$Tone Ring$,$Attack$}
			\end{axis}
		\end{tikzpicture} }
	} &
	\subfloat[ ]{\label{fig:timelapse2}	
		\resizebox{0.21\textwidth}{!}{\begin{tikzpicture}
			\begin{axis}[xlabel={Freq(Hz)},ylabel={Residual $G_k$},xmin=30, 	xmax=150,ymax = 120,axis lines={center},]
				\addplot coordinates {(48,0) (50,0)    (51,0)   (52,0) (55,0)};
				\addplot coordinates{(117,0)	(120,0)    (121,0)    (122,0) (125,0)};
				\legend{$Tone Ring$,$Attack$}
			\end{axis}
		\end{tikzpicture} }
	}&
	\subfloat[ ]{\label{fig:timelapse2}	
		\resizebox{0.21\textwidth}{!}{\begin{tikzpicture}
			\begin{axis}[xlabel={Freq(Hz)},ylabel={Residual $G_k$},xmin=30, 	xmax=150,ymax = 120,axis lines={center},]
				\addplot coordinates {(48,0) (50,0)    (51,0)   (52,0) (55,0)};
				\addplot coordinates{(117,0)	(120,12)    (121,25)    (122,8) (125,0)};
				\legend{$Tone Ring$,$Attack$}
			\end{axis}
		\end{tikzpicture} }
	}&
	\subfloat[ ]{\label{fig:timelapse2}	
		\resizebox{0.21\textwidth}{!}{\begin{tikzpicture}
			\begin{axis}[xlabel={Freq(Hz)},ylabel={Residual $G_k$},xmin=30, 	xmax=150,ymax = 120,axis lines={center},]
				\addplot coordinates {(48,0) (50,0)    (51,0)   (52,0) (55,0)};
				\addplot coordinates{(117,0)	(120,53)    (121,67)    (122,35) (125,0)};
				\legend{$Tone Ring$,$Attack$}
			\end{axis}
		\end{tikzpicture} }
	}
	\end{tabular}
}
\caption{\label{fig:timelapse} (a) The amplitude of the $N$-point DFT representation of the measured signal. Energy is concentrated around two frequencies, one of which corresponds to the original tone ring frequency while the other corresponds to the attack signal. (b-d) An illustration of the frequency-domain error for an attacked signal evolving over time: (b) shows the errors for all frequency components at the beginning of a ``silent'' period, where the errors are initialized at 0; (c) shows the errors after some time, showing the increase in error on the attacked frequencies; (d) shows that the error of attacked frequencies continues to increase until the silence period ends or the attack ceases.  }
\vspace{-3mm}
\end{figure*}

\subsection{$\chi^2$ PyCRA Resilient Estimator}
\label{sec:sub:resilient_chi2}
We can extend this simple resilience estimator to a more formal $\chi^2$ estimator as was done for the attack detection scheme.  Because the resilience estimator operates in the frequency domain, we need to extend the model in Section \ref{sec:pycra} to the frequency domain as well. In particular, we use the recursive formulation of the Discrete Fourier Transform (DFT) which updates the previously calculated $N$-point DFT with the information of the most recently measured signal. The recursive DFT takes the following form:
\begin{align}
	Y_k(t+1) &= e^{(j2\pi k/N)} Y_k(t) + e^{(-j2\pi k(N-1)/N)} y(t) - e^{(j2\pi k/N)} y(t-N)
\end{align}
where $Y_k(t) \in \Ce^N$ is the $k$th component of the $N$-point DFT of the sensor output $y(t)$ at time $t\in \N_0$.

Similar to the detection scheme described in Section \ref{sec:sub:ch2detector}, we design a $\chi^2$ detector operating in the frequency domain in order to detect which frequency-concentrated energy peaks have been attacked. The $\chi^2$ detector uses the sensor model along with the recursive DFT formula to predict the natural response of the tone gear  (in the case of no attack) denoted $\widehat{Y}_k(t)$. We define the $\chi^2$ residual in the frequency domain as $Z_k(t) = \vert Y_k(t) \vert - \vert \widehat{Y}_k(t) \vert$. At the end of the silence time, we calculate the residual over the window of length $T$ as:
\begin{align}
G_k(t) = \frac{1}{T}\sum_{\tau=t-T+a}^{t} Z_k^2(\tau)
\end{align}

Finally, we set an alarm trigger $\beta$ against which we compare the value of $G_k(t)$ for all $k$ frequency components. Whenever the condition $G_k(t) > \beta$ is satisfied, we declare that the frequency $k$ is under attack. Again, the value of $\beta$ must be selected based on the noise information embedded in the model. 

 
\section{Detecting Eavesdropping Attacks using PyCRA}
\label{sec:sniffing}
The basic operation of any sensor (active or passive) requires transduction of energy from some medium (heat, acoustic, optical, magnetic, etc.) to an electrical signal. From the Law of Conservation of Energy, this transduction and therefore the act of sensing itself necessarily removes energy from the system, effectively modifying the very signal being measured. In the case of active sensors, this affects the energy emitted by the actuator and results in an attenuated signal observed by the sensor.  This effect is the basis for such technologies as RFID and other backscatter communication, where the attenuation is changed over time as a method for encoding data.  Therefore, if the following condition is satisfied:
\begin{itemize}
\item[\textbf{AD1}] The interaction between the eavesdropping sensor and the measured signal is significant enough to cause the measured signal to deviate from the model.
\end{itemize}
then, PyCRA can be used as a detection mechanism for passive sniffing attacks, even though there is no malicious \emph{signal}---i.e., no external source of energy in the system as shown in Section \ref{sec:resultsRFID}.


\section{The Confusion Phase: Another Fundamental Limitation}
\label{sec:theory}

Every physical signal is subject to random perturbations i.e., noise. A fundamental characteristic of this noise is the \emph{signal to noise ratio} (SNR). This SNR determines the ability of any sensor to distinguish between changes in a signal of interest and the random noise. As with the physical delay $\tau_{attack}$, this SNR is fundamental, and it is never equal to zero. As a result, if a signal is within the noise floor (less than the noise amplitude), it is fundamentally impossible to detect any change in the physical signal~\cite{VeerValliGQPD2005}. The purpose of this section is to show how PyCRA can use this fundamental and immutable constraint on how quickly an attacker can detect changes in the physical challenge in order to  introduce additional delay on the capability of the attacker to respond to physical attacks and hence
enhance both the detection as well as the resilience  performance.
We refer to this strategy as the \emph{confusion phase}.

\subsection{Confusion Phase}

We begin by presenting the sequence of actions performed by PyCRA, as depicted by Figure \ref{fig:raw_signal}.\\
\textbf{[Step 1]} In steady state PyCRA actuates a constant signal of amplitude $A$.\\
\textbf{[Step 2]} At a random time $t_{challenge}$
PyCRA issues a physical challenge by entering the silent phase.\\
\textbf{[Step 3]} PyCRA draws at random the silent period length $\Ga$, according to a probability function with a certain decay rate.\\
\textbf{[Step 4]} At the end of the silent period PyCRA turns the signal back on; however, this time the amplitude of the constant magnetic field is set to $\frac{A}{\beta}$, where $\beta > 1$.\\
\textbf{[Step 5]} After a random time $t_{confusion}$ PyCRA increases the amplitude of the magnetic field back to $A$.

Recall that one of the attacker's goals is to remain stealthy (Goal G1). If the attacker is unable to instantaneously detect the changes in the physical challenge, he or she will reveal themselves. Due to the existence of noise, no attacker---whether using software or hardware to counter the physical challenges issued by PyCRA---can instantaneously detect the change in the physical challenge. That is, there always exists a non-zero probability of the attacker missing the changes in the physical challenge. In this section, we detail a theoretical result that explains the relationship between the amplitude of the physical challenge within the confusion phase and the probability that the attacker will fail to detect changes in the physical challenge.

In the remainder of this section, we analyze---from the attacker's point of view---how this sequence of actions guarantees with a very high probability a finite delay on the attacker side. This delay is independent of the hardware or software employed by the attacker. Hence, the length of the delay can be adjusted to fit the response time of PyCRA such that it can detect the attack with a very high probability. Furthermore, introducing a delay can actually improve the system's ability to recover the system from the attack as discussed in Section~\ref{sec:ext_resilience}.

\begin{figure}[!t]
\centering{
	\includegraphics[width=0.65\columnwidth]{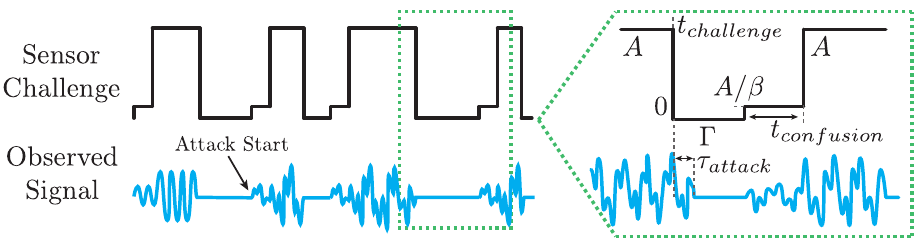}
}
\caption{\label{fig:raw_signal}Sensor actuator output (top) with confusion and silence phases and the corresponding raw signal (bottom) with an attack.} 
\vspace{-4mm}
\end{figure}

\subsection{Bayesian Quickest Change Point Detection: Basic Definitions and Results}\label{subsec:QCPDBasicDefinitionResults}
In this subsection we investigate from the attacker's point of view the fundamental limitations on its ability to detect when the physical challenge has been abandoned so that he can start re-attacking the sensor. Towards this end, we consider the Bayesian quickest change detection problem~\cite{VeerValliGQPD2005} which shows that SNR introduces immutable constraints on the capability of the attacker to perform such detections. 

In the Bayesian quickest change detection problem~\cite{VeerValliGQPD2005} a sequence of  random variables
$X_{1},X_{2},\dots$ which are i.i.d are observed, until a change occurs at an unknown point in time $\Ga\in\pac{1,2,\dots}$ after which the observations are again i.i.d but according to a different distribution. In other words, we consider the case when $X_{1},\dots,X_{\Ga-1}$ are drawn at random according to the Probability Distribution Function (PDF) $f_{0}\pa{x}$, whereas $X_{\Ga},X_{\Ga+1},\dots$ are drawn according to the PDF $f_{1}\pa{x}$.

The goal is to detect the change with minimal average delay based on the observed data and the change point distribution, subject to a constraint of false alarm probability. Therefore, a procedure decides whether a change has occurred at point $\tau$ based on the observations $X_{1},\dots,X_{\tau}$. For simplicity we denote by:
{\small\beq{eq:BD_taudefinition}{\tau\pa{X_{1},\dots,X_{\tau}}}}
the function determining whether a change occurred at time $\tau$ based on observing the data $X_{1},\dots,X_{\tau}$. A false alarm event occurs if one falsely detects a change and hence reports  a time $\tau$ that is less than $\Ga$ (i.e., $\tau < \Ga$). Recall that  the attacker's goal is to ensure that its attack signal is concealed.

Therefore, a false alarm (for the attacker) happens when the attacker erroneously concludes that the sensor is no longer in the silent or confusion phase. In this case the attacker could reveal itself by injecting its signal. On the other hand if the attacker detects the change with a delay, then for this period  the active sensor obtains uncorrupted measurements.
These uncorrupted measurements can be used then to enhance the detection and resilience of PyCRA as we show in Section~\ref{sec:usingDelay}. For the remainder of this section,
 we use the following notation to denote the detection delay at the attacker side:
{\small\beq{eq:BDdelaydef}{\pa{\tau-\Ga}^{+}}}
where $\pa{x}^{+}$ equals $x$ for $x>0$ and zero otherwise. Our objective is to study this detection delay. In particular, we will show that because of the SNR this delay always exists and never equals to zero. Moreover, we show how PyCRA can increase this detection delay to enhance its performance.

Mathematically, the quickest change detection problem is formalized as the following minimization problem:
{\small\beq{eq:BDBQCDProblem}{\min_{\tau\in\Delta\pa{\Al}}E\pac{\pa{\tau-\Ga}^{+}}}}
subject to:
{\small\beq{eq:BDFAProblem}{P_{FA}=Pr\pa{\tau<\Ga}\le\Al\quad \forall\tau\in\Delta\pa{\Al}}}
In other words, the Bayesian quickest change detection problem finds the minimal expectation of the detection delay under the constraint that the false alarm probability $P_{FA}$ is equal to or smaller than $\Al$. 
\subsection{Characterizing the Asymptotic Delay of the Attacker}\label{subec:AsymptoticDelayAttacker}
In this subsection, we argue that if the attacker would like to have a small false alarm rate, then it will suffer from high detection delay. Intuitively, if the attacker wants to reduce probability of detection (and hence selects a small false alarm probability), then the attacker needs a longer delay to average over the ambient random noise. This intuition is captured by the following  result.
\begin{theorem}\label{th:DelayGuarantee}
\footnote{Although we present the proof for the case of AWGN, the equality holds for any sequence of random variables that are drawn i.i.d according to $f_{0}\pa{x}$, $f_{1}\pa{x}$ for which there exist $g_{1}\pa{\Al}$, $g_{2}\pa{\Al}$ such that:  
(1) $\lim_{\alpha\to 0}\ln\pa{f_{1}\pa{g_{1}\pa{\alpha}}}/\ln\pa{1/\alpha}= 1$; (2) $\lim_{\alpha\to 0}\ln\pa{\frac{f_{1}\pa{g_{1}\pa{\alpha}}}{f_{0}\pa{g_{1}\pa{\alpha}}}}/\ln\pa{1/\alpha}=0$; (3) $\lim_{\alpha\to 0}\ln\pa{g_{2}\pa{\alpha}}/\ln\pa{1/\alpha}=0$; (4)  $\lim_{\alpha\to 0}\ln\pa{Pr\pa{\Gamma=g_{2}\pa{\alpha}}}/\ln\pa{1/\alpha}\ge 1$.
In a nutshell $g_{2}\pa{\Al}$ allows to consider a finite number of elements, that scales with $\Al$, whereas $g_{1}\pa{\Al}$ enables to show that the elements converge uniformly over as $\Al$ decreases.} 
Consider an attacker attempting to detect changes in a physical challenge signal, subject to a false alarm probability $\Al$.
For any strategy the attacker chooses, and because of the signal to noise ratio (SNR) which exists for any sensor, the probability of the attacker having a constant detection delay $\tau > 0$ is bounded away from zero. Moreover, when $\Al\ll 1$ the probability of a delay smaller than any constant $K$ fulfills:
\beq{eq:ProbDelaySmallAlpha}{Pr\pa{\Ga\le\tau\le \Ga + K}\dot{=}\Al}
for any $\tau\in\Delta\pa{\Al}$ (where $g\pa{\Al}\dot{=}f\pa{\Al}$ when $\lim_{\Al\to 0}\frac{\ln\pa{g\pa{\Al}}}{\ln\pa{\Al}}=\frac{\ln\pa{f\pa{\Al}}}{\ln\pa{\Al}}$).
\end{theorem}
\begin{proof}
The full proof is in Appendix \ref{sec:proofodDetDelayDecay}.
\end{proof}
In other words, if an attacker wants to detect the end of the physical challenge within one time step (e.g., $K = 1$) while the false alarm probability is small (e.g, $\Al = 10^{-9}$), then Theorem~\ref{th:DelayGuarantee} guarantees that the probability of the attacker achieving this objective is equal to $10^{-9}$. The higher false alarm probability the attacker chooses, the higher the probability of it to being able to detect the end of the physical challenge period. Therefore, Theorem~\ref{th:DelayGuarantee} gives the attacker tradeoff between false alarm and detection delay.

Next, we show that due to the fact that PyCRA performs multiple physical-challenges over time, the attacker has to set his  false alarm probability to a small number  in order to avoid being detected over time.
\begin{cor}\label{cor:multipleChRe}
Assume the attacker has a false alarm probability $\Al$ and that PyCRA performs $K$ physical-challenges over time. The probability of detecting the attack is larger than or equal to $1-\pa{1-\Al}^{K}$.
\end{cor}
\begin{proof}
The probability of detecting an attack at each instance is lower bounded by the misdetection probability. Therefore, the probability of detecting an attack at the $l$th physical challenge is lower bounded by the geometrical distribution $\pa{1-\Al}^{l-1}\cdot \Al$. Hence, the probability of detecting an attack after $K$ physical-challenges is lower bounded by
$\sum_{i=0}^{K-1}\Al\cdot \pa{1-\Al}^{i}=1-\pa{1-\Al}^{K}.$
\end{proof}

From Corollary \ref{cor:multipleChRe} we get for $\Al\ll 1$ and $K\cdot\Al\ll 1$ that the probability of detecting an attack after $K$ physical-challenges is larger than approximately $K\cdot \Al$. 
Therefore, in order to maintain a small detection probability over time, the misdetection probability has to fulfill $\Al\ll 1$.

\subsection{Using Delay for Detection and Estimation} \label{sec:usingDelay}
In this subsection, we present a method to increase the detection delay induced by the confusion phase, as well as the interplay between increasing the detection delay and the capability of PyCRA to perform the resilience estimation discussed in Section~\ref{sec:ext_resilience}. We begin by proving a theorem showing that by decreasing the
actuated amplitude (in the confusion phase) by a factor of $\beta>1$, the detection delay on the attacker side $K$ can be increased by a factor of $\beta^{2}$, with the same probabilistic guarantee for the delay. Then we show that if PyCRA uses a Maximum Likelihood (ML) estimation procedure to perform the resilient estimation, then increasing the detection delay on the attacker side assists in enhancing the resilience of PyCRA.

We start by proving the following theorem showing the relation between decreasing the amplitude of the actuated signal and the probability of a constant delay.
\begin{theorem}\label{th:ExtendedDelay}
Consider an attacker attempting to detect changes in a physical challenge signal with misdetection probability $\alpha$. 
For any strategy the attacker chooses, and because of the SNR exists at any sensor, the probability of the attacker having a constant detection delay $\tau > 0$ is bounded away from zero,
i.e., with high probability the attacker will detect a change and turn off his signal only after time $T$ after the beginning of the confusion period.
In addition, decreasing the amplitude of the signal emitted by the active sensor during the confusion period by a factor of $\beta >1$ increases the delay $\tau$ by a factor of $\beta^{2}$.

Mathematically, the following equality holds:
{\small\beqn{}{Pr\pa{\Ga\le\tau\le \Ga + \beta^{2}\cdot K|\frac{A}{\beta},\beta^{2}T}=Pr\pa{\Ga\le\tau\le \Ga + K|A,T}}}
where $Pr\pa{\Ga\le\tau\le \Ga + K|A,T}$ is the probability of a delay of length smaller than or equal to $K$ when actuating with an amplitude $A$ and drawing the delay over a grid with a period time $T$, and $Pr\pa{\Ga\le\tau\le \Ga + \beta^{2}\cdot K|\frac{A}{\beta},\beta^{2}T}$ is the probability of a delay of length smaller than or equal to $\beta^{2} K$ when actuating with an amplitude $\frac{A}{\beta}$ and drawing the delay over a grid with a period time $\beta^{2} T$.
\end{theorem}
\begin{proof}
Actuating with an amplitude $\frac{A}{\beta}$, changing the period time to $\beta^{2} T$, and projecting the received random process over this period time and then normalizing the projection by a factor of $\beta$ leads to the same pdf as the one that corresponds to actuating with an amplitude $A$, having a period time $T$, and projecting the received signal over this period time.
\end{proof}

Theorem \ref{th:ExtendedDelay} shows that the response time of the attacker and the delay can be decoupled. Furthermore, the delay can be adjusted to suit the response time of PyCRA.

We now discuss how increasing the delay affects the ML estimation and hence the resilience performance of PyCRA. As discussed in Section~\ref{sec:ext_resilience}, we consider the case when both the true sensor measurements as well as the attack signal have the same structure. Without loss of generality,  we can parameterize this signal structure by a parameter $\theta$ and therefore the signal space can be written as $\phi_{\theta}\pa{t}$. For instance, if the signal structure is dominated by sinusoidal functions, then $\phi_{\theta}\pa{t}$ can be written as$\phi_{\theta}\pa{t}=\sin \pa{\theta t}$. Similarly, if the signal structure could be a rectangular function with a period time $\frac{1}{\theta}$ and so on. In general, these functions are asymptotically uncorrelated and hence:
{\small\beq{eq:CorrletaionofwheelReflection}{
\lim_{T\to \infty}\frac{1}{T}\int_{0}^{T}\phi_{\theta}\pa{t}\cdot \phi_{\theta^{'}}\pa{t} dt=\left\{\begin{array}{cc}1 & \theta=\theta^{'}\\
0 & \theta\neq\theta^{'}
\end{array}\right.}}
where it can be assumed for $\theta\neq\theta^{'}$ that the correlation decreases monotonically.
Due to the AWGN, the signal from which PyCRA estimates the true signal measurement:
{\small\beqn{}{
y\pa{t}=\phi_{\bar{\theta}}\pa{t}+n\pa{t}\qquad 0\le t\le T_{Delay}
}}
where $n\pa{t}$ is the AWGN. In this case the ML estimation of $\bar{\theta}$ is:
{\small\beq{}{
\widehat{\theta}=\arg\max_{\theta} Re\pa{\int_{0}^{T_{Delay}} y^{\dagger}\pa{t}\cdot \phi_{\theta}\pa{t}dt}-\int_{0}^{T_{Delay}}|\phi_{\theta}\pa{t}|^{2}dt.
}}
Considering the projection of $y\pa{t}$ on $\phi_{\bar{\theta}}\pa{t}$ leads to the same $\SNR$ when increasing the delay by a factor of $\beta^{2}$ while decreasing the actuated amplitude by a factor of $\beta$ (assuming $\int_{0}^{T_{Delay}}|\phi_{\theta}\pa{t}|^{2}dt$ increases by a factor of $\beta^{2}$ with the delay, which is a reasonable assumption). On the other hand, based on (\ref{eq:CorrletaionofwheelReflection}), we get that the bias resulting from the projection of $y\pa{t}$ on $\phi_{\theta}\pa{t}$ for $\theta\neq\bar{\theta}$ decreases as the delay increases, which is desired for estimation.

In general, increasing the delay might affect the performance of PyCRA while improving the estimation, and decreasing it may lead to a less accurate estimate of the true sensor measurement. Therefore, the delay should be chosen based on the set of reflected functions as well as on the requirements of the system .

Incorporating a few PyCRA challenges together can significantly improve the latter estimation quality, if the signal reflected from the measured entity remains constant across these tests. The later condition can be easily achieved since PyCRA can issue challenges over several orders of magnitude faster compared to the change in the physical signal as discussed in Section~\ref{sec:discussion}.
The next corollary quantifies the improvement as a function of the number of tests.
\begin{cor}
Assume that $N$ PyCRA challenges are incorporated together, and further assume that the signal reflected from the measured entity remains constant across these tests. Also, assume that the period of time between tests is drawn uniformly, and that the maximal period of time between tests is $L\cdot T_{Delay}$. In this case the SNR of the hypothesis incorporating the true sensor measurement , $\bar{\theta}$, increases by a factor of $N$. In addition, when $N \gg 1$ the effect of the bias for $\bar{\theta}\neq \theta$ decreases to:
{\small\beqn{}
{
\int_{0}^{L\cdot T_{Delay}} y^{\dagger}\pa{t}\cdot \phi_{\theta}\pa{t}dt.
}}
\end{cor}
\begin{proof}
The increase in the SNR of $\bar{\theta}$ is straightforward. When $N\gg 1$ the offset between different correlations spreads uniformly and therefore, effectively, we get correlation over a period of $L\cdot T_{Delay}$ instead of $T_{Delay}$.
\end{proof}

The following sections describe three case studies of the PyCRA security scheme.  The first is an active attack detector for magnetic encoder sensors widely used in automobiles and industrial machineries.  The second outlines an algorithm for adding attack resilience in magnetic encoder sensors and the properties that make this resilience feasible. The final case study explores a method for detecting \emph{passive} eavesdropping attacks carried out against RFID systems.


\section{Case Study (1): Detecting Active Spoofing Attacks for Magnetic Encoders}
\label{sec:absattacks}

\tikzstyle{block} = [draw, fill=blue!20, rectangle, minimum height=3em, minimum width=6em]
\tikzstyle{sum} = [draw, fill=blue!20, circle, node distance=1cm]
\tikzstyle{input} = [coordinate]
\tikzstyle{output} = [coordinate]
\tikzstyle{pinstyle} = [pin edge={to-,thin,black}]

Magnetic encoders are active sensors used in a wide array of industrial, robotics, aerospace, and automotive applications.  The goal of an encoder is to measure the angular velocity or position of a gear or wheel in order to provide feedback to a motor controller.  
The operation of these systems depends heavily on the accuracy and timeliness of the individual encoders.  This section describes the basic operation of magnetic encoders in particular and the types of attacks that can be mounted against them.

\begin{figure}
\centering
\includegraphics[width=0.45\columnwidth]{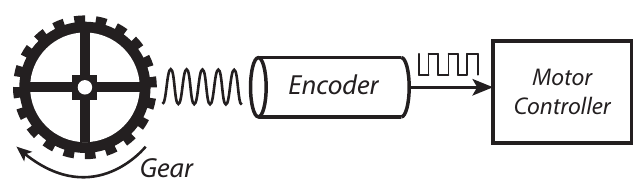}
\caption{\label{fig:abscircuit}Flow diagram for a typical magnetic encoder:  The signal begins as a reflected magnetic wave from a gear. This signal is captured by a pick-up coil or Hall Effect sensor, conditioned into a clean square wave, and finally translated into an angular velocity.}
\vspace{-3mm}
\end{figure}

\subsection{Magnetic Encoders}
\label{sec:sub:speedsensors}

Magnetic encoders rely on magnetic variations to measure the angular velocity of a gear or wheel and are often  designed to handle dust, mud, rain, and extreme temperatures without failing.  
The goal of each encoder is to provide a signal whose frequency corresponds to the speed of a gear. These signals are conditioned and passed to a motor controller unit which detects if any corrective actions need to be taken.   

Typical magnetic encoders operate by generating a magnetic field in the presence of a rotating ferromagnetic \emph{tone ring} or \emph{tone wheel}.  This ring has a number of teeth on its edge so that the reflected magnetic wave as observed by the encoder varies over time as a (noisy) sinusoidal wave. By measuring the frequency of this reflected signal over time, each sensor and consequently the motor controller is able to infer the angular velocity of any given gear, wheel, or motor as illustrated in Figure \ref{fig:abscircuit}. 

Attacks on magnetic encoders have been studied in \cite{YasserABS} in the context of Anti-lock Braking Systems in automotive vehicles. Both simple spoofing [T2] and advanced spoofing [T3] attacks are shown to influence the vehicle stability. In this case study, we show how PyCRA can detect the existence of such attacks.

\subsection{Constructing the PyCRA-secured Magnetic Encoder}
\label{sec:hardware}

\begin{figure}[!t]
\centering
{
	\begin{tabular}{c|c}
	\subfloat[ ]{\label{fig:secureSensor}
		\includegraphics[width=0.3\columnwidth]{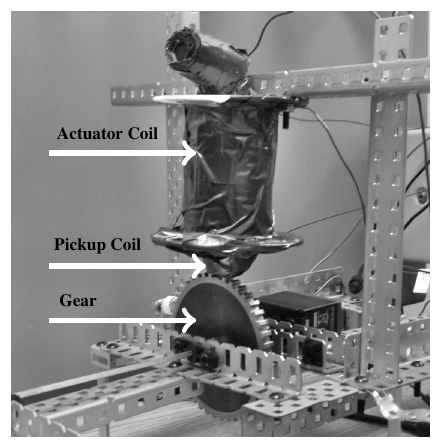}
	}
	&
	\subfloat[ ]{\label{fig:systemId}	
		\includegraphics[width=0.445\columnwidth]{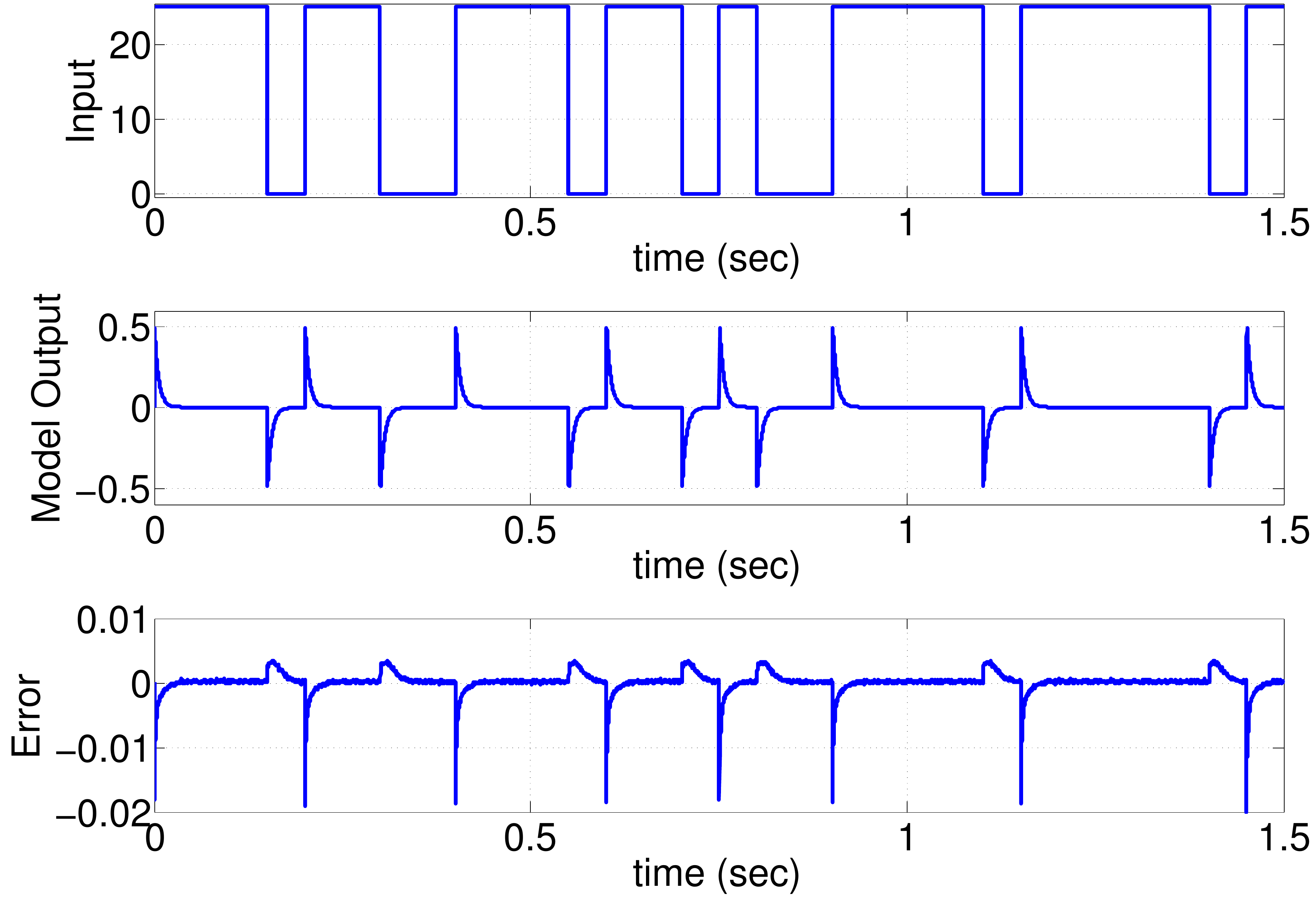}
	}
	\end{tabular}
}
\caption{\label{fig:mag_sensor}(a) PyCRA encoder actuator coil, sensor, and gear setup, (b) Validation of the physical model obtained for the secure sensor. The top figure shows the input fed to the validation phase. The middle figure shows the response, and the bottom figure shows the corresponding error.}
\vspace{-3mm}
\end{figure}

Physically, the proposed secure magnetic encoder sensor consists of two main parts: (i) the front-end   containing the actuator and pickup coils responsible for both probing the rotating tone ring and measuring the response, and (ii) the processing backend. Figure \ref{fig:secureSensor} shows the front-end of the sensor used in our evaluation. The actuator coil depicted is much larger than would be required in a commercial product, because it consists of a magnetic core and a hand-wound high-gauge wire. The following is an overview of the main blocks of the sensor.

\subsubsection{Actuator Coil}
The main component required for the secure sensor is the actuator coil. In this work, we use an insulated copper wire wrapped around a ferromagnetic core and driven using a power amplifier.

\subsubsection{Pickup and  Filtering}
The pickup (measurement) coil is wrapped around the same ferromagnetic core used for the actuator coil. In order to reduce the effect of noise from other EMI sources within the vehicle body, the output of the pickup coil is connected to a differential amplifier with high common-mode rejection. The output of this differential amplifier is connected to the digital processing backend.

Another security concern of the magnetic encoder is the wires connecting the coils to the digital backend. These wires pose a potential vulnerability, as an attacker can cut them and connect his attack module directly. However, such attacks are already accounted for in many systems as addressed in Assumption A1. 

\subsubsection{Processing Elements}
The secure sensor requires enough processing power to perform the necessary computations in real-time. The DSP calculations take place on a high power ARM Cortex (M4 STM32F407) processor, which has ample floating point support. We do not consider any power consumption issues in our design.


\subsection{Obtaining the Sensor Model}
\label{sec:systemID}
The dynamics of the sensor (including the actuator, high gain current amplifier, sensors, and the signal conditioning circuit) are identified using standard system identification methods \cite{ljung1998system}. That is, we applied four different pseudo random binary sequences (PRBS) to the system, collected the output, and then applied subspace system identification techniques in order to build models of increasing complexity \cite{ljung1998system}. Finally we used both whiteness tests and correlation tests to assess the quality of the obtained model \cite{Landaue_Book}.  

In order to validate the model, we generated a random sequence similar to those used in the real implementation of the sensor. We fed the same input to both the sensor and the model and recorded the error. Figure \ref{fig:systemId} shows the results of this experiment. The results show that the model accuracy depends on whether the input changes from one to zero or from zero to one. This difference in the response reflects the behavior of the high gain current amplifier which has different spikes depending on the changes of the input. 

In our case, we are interested only in the case when the sensor enters the ``silence'' time, i.e. $u(t)$ transitions from 1 to 0. In such cases, Figure \ref{fig:systemId} shows that the model is reasonably accurate with an error in the range of $0.005$ Volts.

\subsection{Testbed}
In order to test the PyCRA-secured magnetic encoder, we constructed a testbed consisting of the proposed secure sensor attached to a Mazda Rx7 tone ring. The tone ring is attached to a DC motor which simulates a rotating wheel. An additional coil is added to simulate the effect of an attacker. The attacker coil is also controlled by a high gain amplifier controlled through a real-time xPC Target system connected to MATLAB.

A Mazda RX7 magnetic encoder sensor is also attached to the same tone ring in order to provide ground truth. The output of this sensor is connected to a MAX9926U evaluation kit which includes an interface capable of converting the raw sinusoidal wave into the encoded square wave as shown in Figure \ref{fig:abscircuit}. The output of the proposed secure sensor as well as the output of the MAX9926U is monitored by the same real-time xPC Target for comparison.

\subsection{Calibration against natural variations}
Sensor modeling is usually done in a controlled environment. However, once the sensor is placed in a testbed, multiple natural variations, mechanical asymmetries, and other environmental factors degrade the accuracy of such models. 
To account for these variations, we use a simple learning mechanism to estimate the noise level in the measurements and the deviation between the expected outputs (as per the model) and the actual outputs. Once these parameters are learned, we can set the alarm threshold accordingly.
Results can be further improved by considering online identification-and-calibration of the sensor model.

\subsection{Attack Detection Results for Magnetic Encoders}
We begin with a simple spoofing attack [T2] in which an attacker injects a sinusoidal wave of varying frequency.  Spoofing attacks of this nature attempt to overpower the true frequency of the system and force the sensor to track the false frequency (mirroring the simplistic spoofing attack in \cite{humphreys:gps}). In this experiment, the original tone ring frequency is fixed at 71 Hz and the frequency of the attacking coil increases linearly from 60 Hz to just over 400 Hz. 

As per our attacker model in Section \ref{sec:attacker_model}, we assume that the attacker attempts to conceal his or her presence (Adversarial goal [G1]).  This means that the adversary will be able to detect when the actuator coil is turned off and will, after some time $\tau_{attack}$, temporarily halt the attack. 

The stealthiness of the attacker necessitates that the PyCRA detection scheme have high accuracy even when the attacker is quick to react.  Figure \ref{fig:raw_signal} shows an example of the random PyCRA challenges and the corresponding observed signal both before and after an attack is present.  In this case, the adversary quickly disables the attack after the actuator coil transitions from 1 to 0, lagging only by the small delay $\tau_{attack}$, imperceptible in the figure. We evaluated the PyCRA detection scheme across a range of $\tau_{attack}$ values, $\chi^2$ detection thresholds ($\alpha$), and sampling frequencies ($F_s$).  Note that in order to simulate an attacker with 0 ms physical delays (which is physically impossible), we gave the attacker access to the random signal generated by PyCRA so that the attacker can start shutting down his actuators \emph{before} PyCRA generates the physical challenge.

In total, we conducted over 30 experiments on our experimental testbed to validate the robustness of the proposed security scheme. The resulting accuracy with $F_s = 10$ kHz is depicted by the ROC{\footnote{A Receiver Operating Characteristic (ROC) is a visual aid for evaluating the accuracy of binary classifiers in terms of both true and false positive rates.}} curves in Figure \ref{fig:fixed_roc} for a range of $\alpha$.  From this figure it is clear that between $\tau_{attack} = 500$ and 700 $\mu$s is all that is necessary for PyCRA to accurately distinguish attacked signals from normal signals, if $\alpha$ is chosen appropriately.  With $\alpha$ set to a predetermined value, we can vary $F_s$ as shown in Figure \ref{fig:detection_all}\footnote{The $F_1$ score is a statistical measure of a binary classifier that measures the classifier accuracy in terms of precision and recall.}.  These results show that increasing $F_s$ from 10 kHz to 30 kHz reduces required time for detection to between $\tau_{attack} = 100$ and 200 $\mu$s.  
\begin{figure}[!t]
\centering
{
	\begin{tabular}{c|c}
	\subfloat[ ]{\label{fig:fixed_roc}
		\includegraphics[width=0.4\columnwidth]{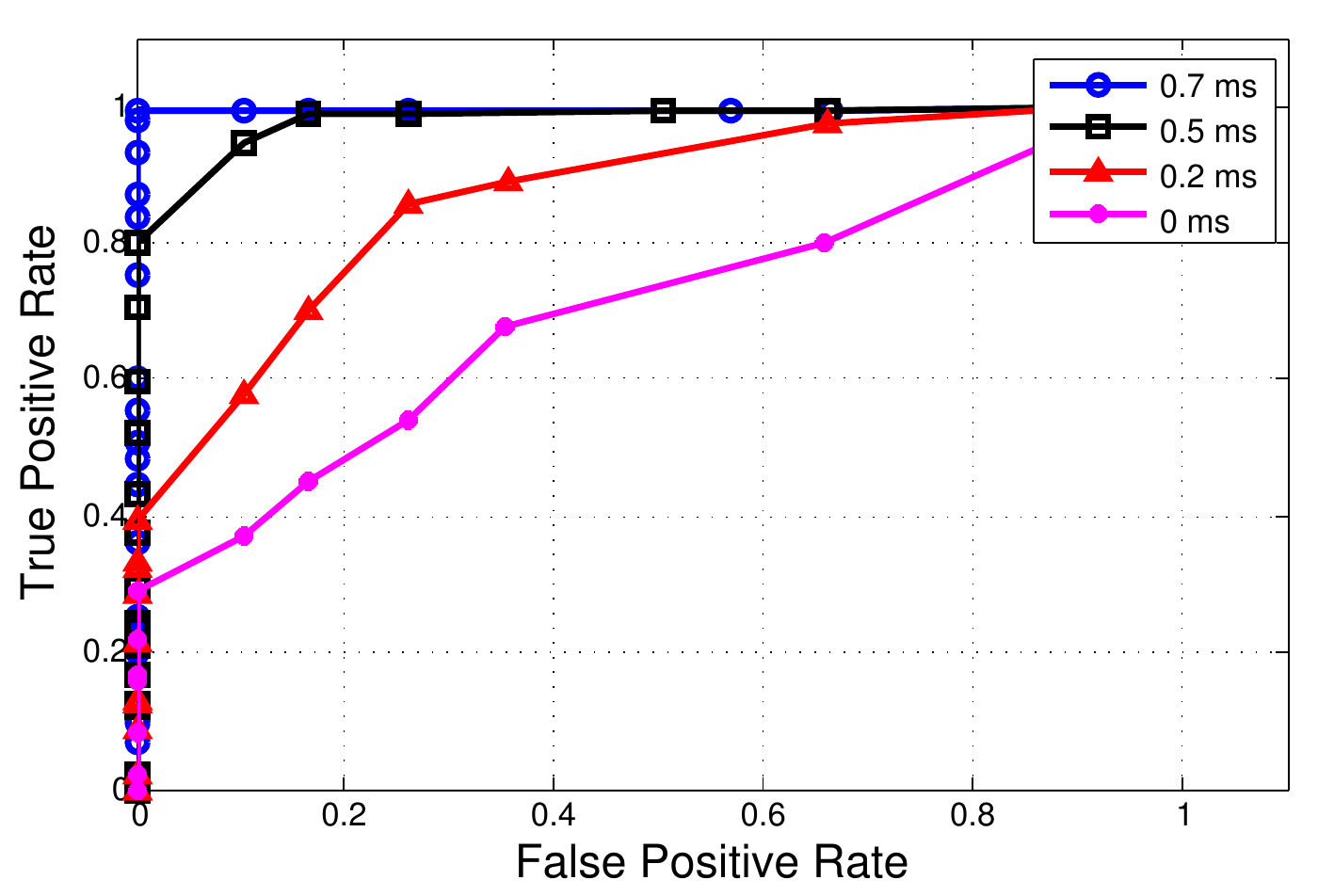}
	}
	&
	\subfloat[ ]{\label{fig:detection_all}	
		\includegraphics[width=0.4\columnwidth]{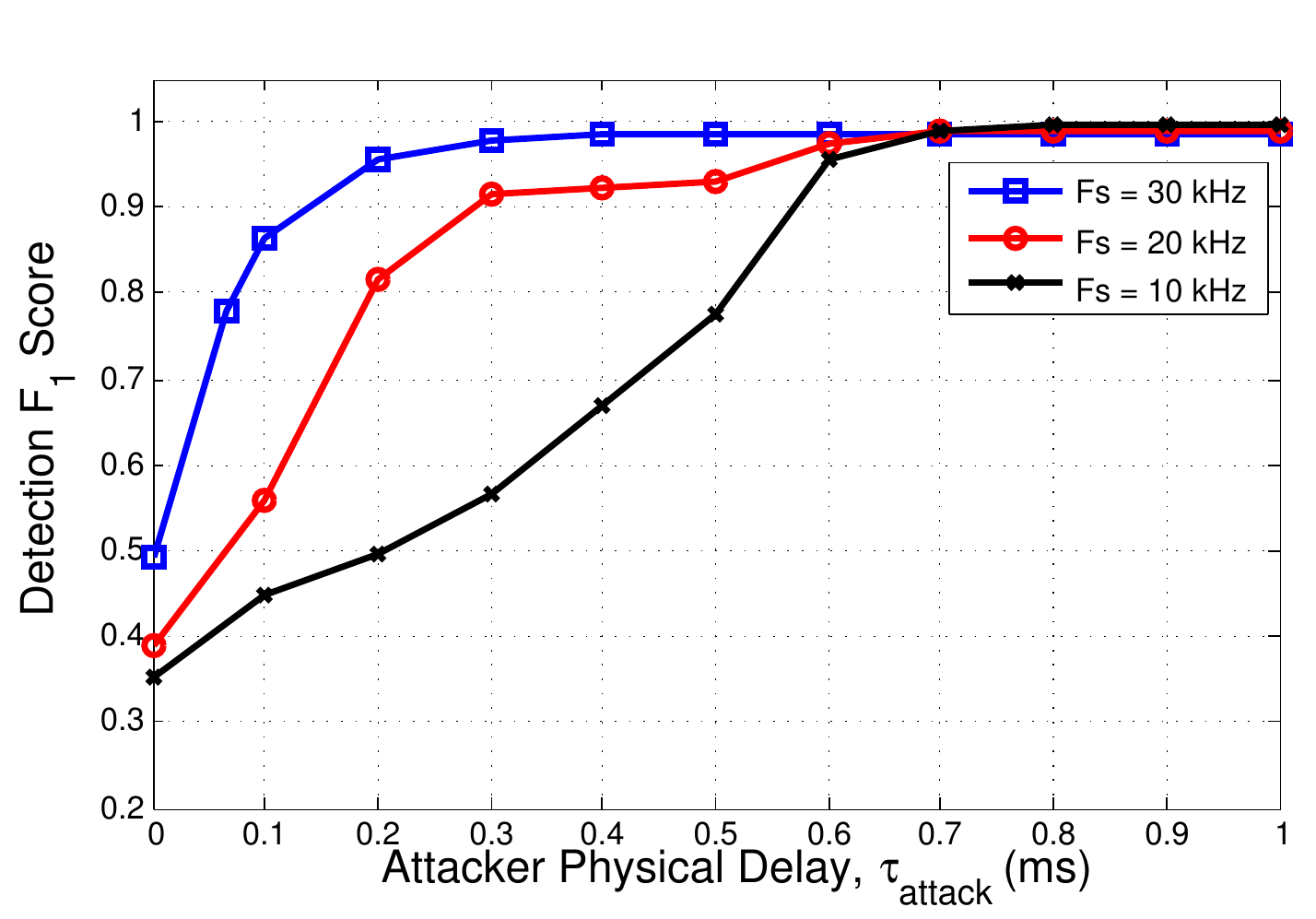}
	}
	\end{tabular}
}
\caption{\label{fig:detection}(a) Accuracy of attack detection for a simple spoofing attack with sampling rate $F_s = 10 $ kHz and a range of $\tau_{attack}$, (b) Attack detection accuracy as a function of $\tau_{attack}$ for several sampling rates, $F_s$.}
\vspace{-4mm}
\end{figure}
Repeating these experiments for the \emph{advanced spoofing attack} [T3] yields similar results.  In fact, there is no fundamental difference between the two in terms of attack detection; this is governed by the dynamics of the attacker's actuator rather than the nature of the attack itself. 

It is important to evaluate this detection accuracy (which is our security guarantee) in terms of the physical delay property $\tau_{attack}$ of the attacker model. 
In practice, the state-of-the-art in low-dimension, high Q-factor inductive coils that provide enough power to carry out a spoofing attack will have $\tau \gg 200 \mu$s{\footnote{These values were obtained by surveying a range of state-of-the-art, commercially available components.}}. From Figure \ref{fig:detection_all} it is apparent that PyCRA has good performance for this range of practical physical delays.

Moreover, the results we have shown thus far use a relatively low sampling frequency (high end micro controllers can operate in the range of 200 kHz).  As illustrated by Figure \ref{fig:detection_all}, higher sampling rates result in reduced attack detection times. However, using low sampling frequencies in our case study serves to illustrate the efficiency of the proposed detection mechanism.

\section{Case Study (2): Resilience to Active Spoofing Attacks for Magnetic Encoders}
\label{sec:absattacks2}

\begin{figure*}
\centering{
\resizebox{0.6\textwidth}{!}{
	\includegraphics{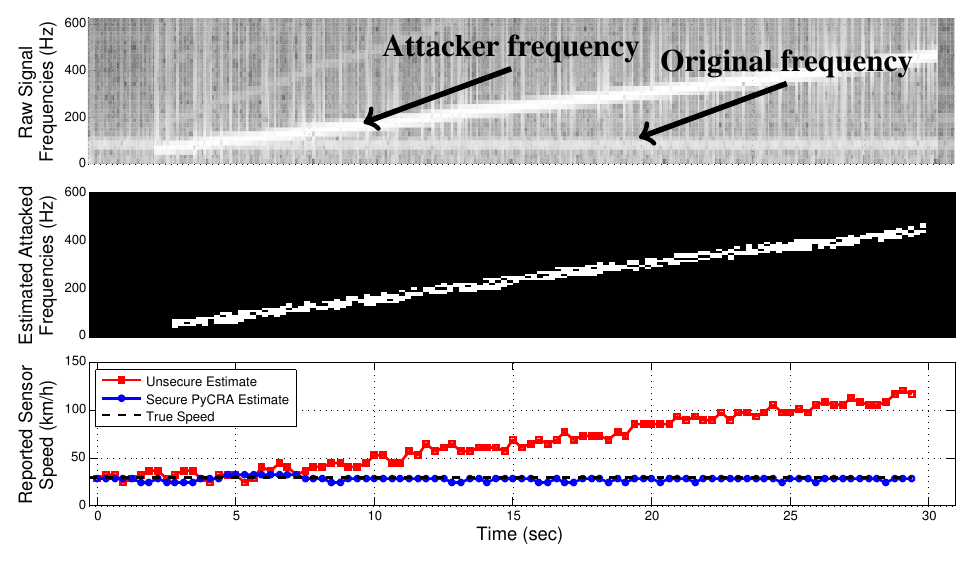}
}
\caption{\label{fig:resilience} A visual tour of the PyCRA resilience scheme against the swept frequency attack.  The spectra of the raw signal (top) is compared against those frequencies estimated to be under attack (middle), resulting in the wheel speed estimates shown (bottom). 
}
}
\vspace{-6mm} 
\end{figure*}

\begin{figure*}[!t]
\centering
{
	\begin{tabular}{c|c}
	\subfloat[ ]{\label{fig:resilience_roc}	
		\includegraphics[width=0.37\textwidth]{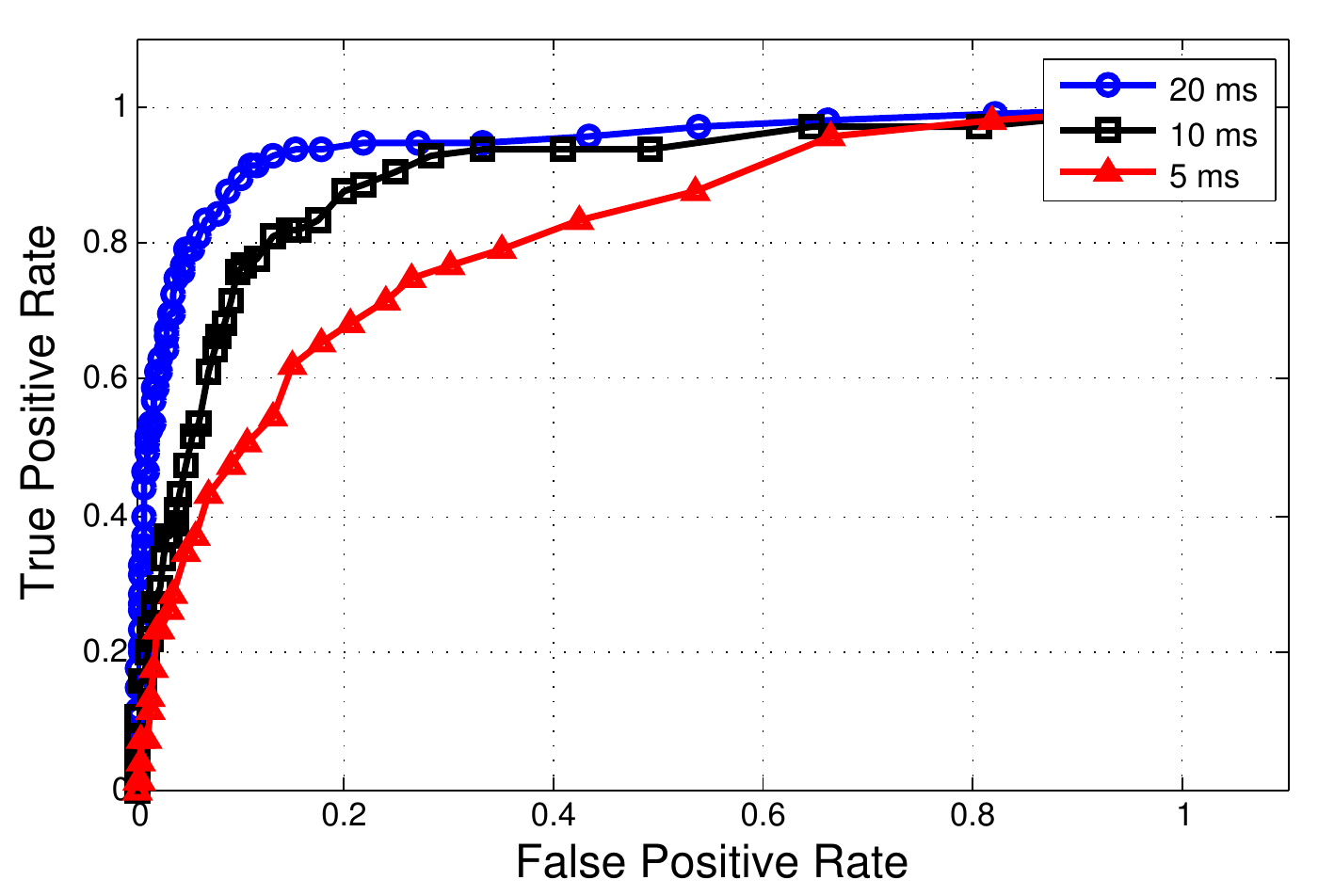}		
	}&
	\subfloat[ ]{\label{fig:resilience_roc30k}	
		\includegraphics[width=0.37\textwidth]{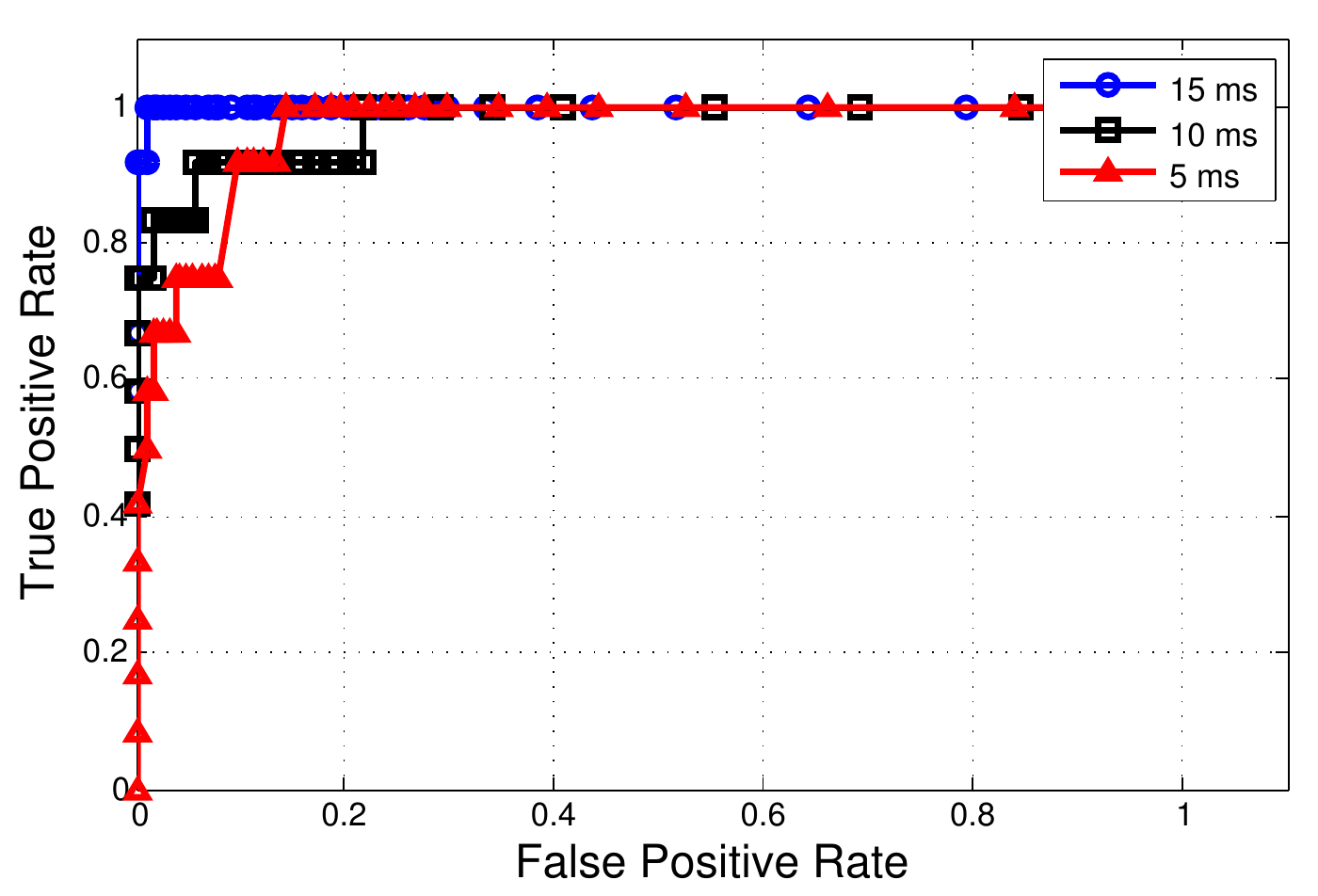}	
	}
	\end{tabular}
}
\caption{\label{fig:roc}Accuracy of predicting the attacked frequencies for a swept frequency attack with a range of $\tau_{attack}$ values for different sampling rates (a) $F_s = 10$ kHz, (b) $F_s = 30$ kHz.}
\vspace{-4mm}
\end{figure*}

In order to demonstrate the PyCRA Resilient Estimator, we implement the scheme described in Section \ref{sec:sub:resilient_chi2} and test against both simple spoofing attacks (consisting of a varying sine wave) and more advanced spoofing attacks that attempt to suppress the true tone ring signal. 
\subsection{Attack Resilience for Simple Spoofing Attacks} We begin by testing the attack resilient sensor against a simple attack consisting of a single sinusoidal waveform swept across multiple frequencies. We recursively compute a series of complex DFT coefficients from the raw input signal shown in Figure \ref{fig:raw_signal}. The magnitude of these DFT coefficients over time is shown in Figure \ref{fig:resilience} (top).  This clearly shows both the true tone ring frequency as well as the swept frequency attack. Second, the errors between the predicted DFT recursions and the measured DFT recursions following every $u(t) = 1\rightarrow 0$ transition are computed in an attempt to capture the spectra of the attacker.  Those frequency coefficients with a residual error greater than $\beta$ are said to be under attack, as shown by the white points in Figure \ref{fig:resilience} (middle).  Finally, at each point in time the candidate frequencies (obtained by peak detection on the spectra of the raw signal) are compared against the set of frequencies estimated to be under attack to arrive at the final secure frequency estimate, shown in Figure \ref{fig:resilience} (bottom).  Note that a traditional encoder sensor would report an erroneous speed as soon as the attack was initiated.  The PyCRA sensor, on the other hand, remains stable throughout the entirety of the attack. 

As with the attack detection scheme, we must evaluate the ability of PyCRA to correctly identify attacked frequencies as a function of $\tau_{attack}$.   In order to do this, we run more than 90 experiments and for each one we compared the known attacker frequencies to the estimated attacker frequencies.  The accuracy of this detection is again depicted in the ROC curves shown in Figure \ref{fig:roc}.  As expected, a considerably longer time is required to accurately identify the attacked frequencies, with 10 to 15 ms sufficing for most cases.

\begin{figure*}
\centering{
\resizebox{0.65\textwidth}{!}{
\includegraphics[width=0.75\textwidth]{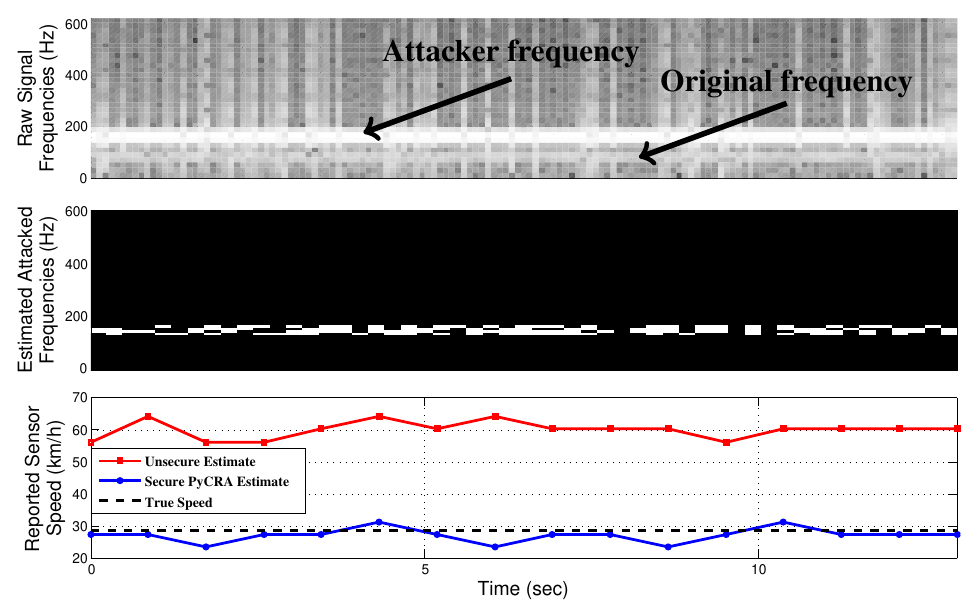}}
}
\captionof{figure}{\label{fig:resilience2} A visual tour of the PyCRA resilience scheme against the spoofing attack.  The spectra of the raw signal (top) is compared against those frequencies estimated to be under attack (middle), resulting in the wheel speed estimates shown (bottom). } 
\vspace{-4mm}
\end{figure*}

\subsection{Attack Resilience against Advanced Spoofing Attacks}
In this section we explore how PyCRA performs in the face of a more advanced spoofing attack---the attacker now attempts to actively suppress the true tone-ring signal by broadcasting a phase-shifted version of the signal and superimposing a spoofed signal of a different frequency.  

Again we set the true tone ring frequency to 71 Hz, but now the attacker attempts to spoof a fixed frequency as shown by the spectra in Figure \ref{fig:resilience2} (top). As before, PyCRA successfully identifies the attacked frequencies (Figure \ref{fig:resilience2} middle) and thus identifies the correct tone ring frequency (Figure \ref{fig:resilience2} bottom).  Because of the physical delay of the attacker, $\tau_{attack}$, the true tone ring frequency appears when the actuator coil transitions from 0 to 1 with increasing magnitude for larger $\tau_{attack}$. This can be seen in Figure \ref{fig:resilience2} (top) as the low amplitude energy centered at 71 Hz. 

We forgo the previous analysis of detection and prediction accuracy as a function of attacker cutoff delay, arguing that the results shown in Figure \ref{fig:roc} apply to advanced spoofing attacks as well.


\section{Case Study (3): Detection of Passive Eavesdropping on RFID}
\label{sec:resultsRFID}

In this section, we discuss  the performance of PyCRA when used for detection of passive eavesdropping attacks on active sensors. In this scenario, an adversary \emph{listens} or \emph{records} the same physical signals captured by the sensor. Indeed this type of attack satisfies assumptions A1-A3 described in Section \ref{sec:pycra} and hence it will be useful to extend PyCRA to such cases.

\subsection{Passive Eavesdropping on RFID}

In this section, we use radio-frequency identification (RFID) as an example where successful passive attacks can have severe consequences.
RFID systems are commonly used to control access to physical places and to track the location of certain objects.  An RFID system consists of two parts: a reader and a tag. The RFID reader continuously sends a magnetic signal that probes for existing tags in the near proximity. Once a tag enters the proximity of the reader, the tag starts to send its ``unique'' identifier to the reader by modulating the magnetic probe signal.

RFID tags can be classified as either passive or active based on the source of their energy. While passive tags rely on the energy transmitted by an RFID reader in order to power their electronics, active tags have their own power supplies. As a result, active tags can be equipped with computational platforms that run cryptographic and security protocols to mitigate cloning attacks \cite{rfid_secret}. On the other hand, passive tags do not enjoy these luxuries and therefore are more prone to cloning attacks. 

Cloning of passive RFID tags can take place in one of two forms. In the first, the attacker uses a \emph{mobile} RFID reader and attempts to place it near the RFID tag. The tag innocently responds to the communication issued by the adversarial reader and sends its unique identifier. The other form of attack carried out against RFID systems is to eavesdrop on the communication between the tag and a legitimate reader.  
RFID protective shields and blockers \cite{rfid_overview,rfid_security} are commonly used as countermeasure to the first form of cloning attacks discussed above.  Unfortunately, these shields are of no use against the second type of attacks, because the user is obliged to remove the protective shield before using the tag with the \emph{legitimate} RFID reader, at which time an adversary can successfully eavesdrop.

To carry out an eavesdropping attack, a sniffing device must be placed in close proximity to the RFID reader so that it can measure the electromagnetic waves transmitted by the reader and reflected by the tag. In the following results, we show how PyCRA is able to detect the existence of such an attack, allowing the reader to disable communication with the tag before revealing private information.

\begin{figure}
  \begin{minipage}[t]{0.4\textwidth}
  \mbox{}\\[-\baselineskip]
    \begin{tabular}{c|c}
	\subfloat[]{\label{fig:rfid_schematic}
		\includegraphics[width=0.22\textwidth]{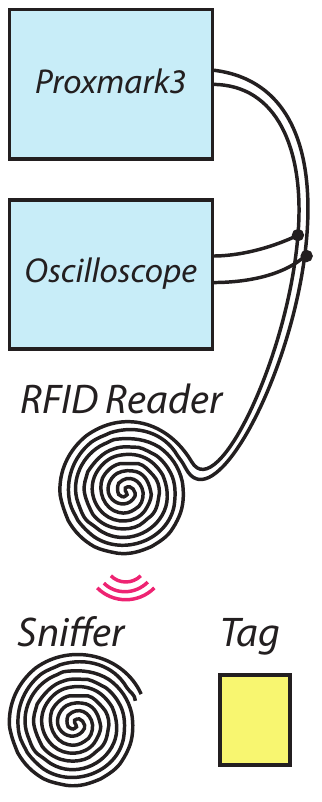}
	}&
	\subfloat[]{\label{fig:rfid_system}
	\includegraphics[width=0.41\textwidth]{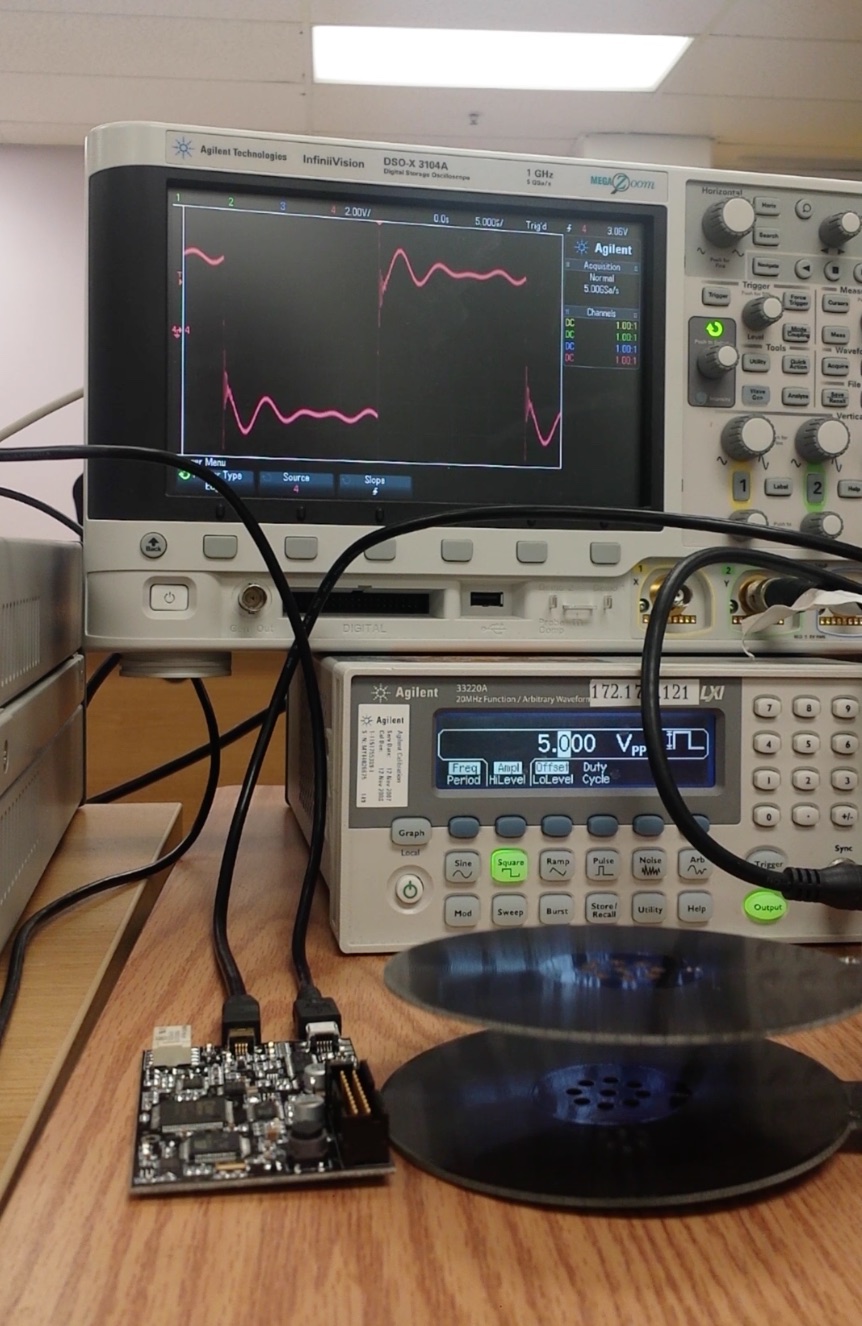}
	}
	\end{tabular}
  \end{minipage}\hfill
  \begin{minipage}[t]{0.5\textwidth}
  \mbox{}\\[-\baselineskip]
    \caption{
       The schematic used in the RFID eavesdropping case study (a) and corresponding hardware setup (b). The setup contains two low frequency antennas (one for the RFID reader and one for the eavesdropper) along with a Proxmark3 board running the PyCRA detection algorithm. The analog signal is also captured by a digital oscilloscope for visualization
    } \label{fig:rfid_system}
  \end{minipage}
\vspace{-3mm}  
\end{figure}

\subsection{Using PyCRA to Detect Eavesdroppers}
Recall from the physics of electromagnetic waves that antennas placed within close proximity will always affect each other's electrical signals to some degree. This is a fundamental law in physics known as \emph{magnetic coupling} \cite{brauer2006magnetic} and is used in the design of RFID. Therefore this attack satisfies Assumption AD1 described in Section \ref{sec:sniffing} and consequently PyCRA can be used in such cases.
Note that, similar to the physical delays, the magnetic coupling assumption is a fundamental limitation that the attacker cannot overcome. Hence, we can use the PyCRA detection algorithm outlined in Section \ref{sec:pycra} by computing the residual between the model (which assumes magnetic coupling only with the RFID tag) and the sensor output which suffers from magnetic coupling with the eavesdropping antenna.  
This is shown in the experimental results in the next subsection.


\subsection{Hardware Setup and Real-time Results}
Figure \ref{fig:rfid_system} shows the hardware setup used to carry out this case-study. In this setup, two identical low frequency RFID antennas are used. The first RFID antenna is used as the \emph{legitimate} RFID reader while the second is used to eavesdrop. We built a PyCRA-enabled RFID reader on top of the open source RFID Proxmark3 board, adding random modulation to the probing signal and constructing the appropriate models as outlined in Section \ref{sec:pycra}. 

Figure \ref{fig:RFID_response_normal} (top) shows the received electromagnetic wave of an RFID reader operating in the classical operational mode. In this mode, the RFID reader generates the standard 125KHz sinusoidal wave that is used to communicate with the RFID tag. Figure \ref{fig:RFID_response_normal} (middle) shows the resulting electromagnetic wave when an eavesdropper uses an identical antenna to listen. In this case it is hard to distinguish between the two waves and hence hard to detect the existence of an eavesdropper. This is better illustrated by the residual between the two waves as shown by the residual in Figure \ref{fig:RFID_response_normal} (bottom).

On the other hand, Figures \ref{fig:RFID_response_pycra_noattack} and \ref{fig:RFID_response_pycra} show the result of the same case study when PyCRA is used with and without an eavesdropper present, respectively. In this mode, the PyCRA algorithm occasionally halts the normal operation of the RFID reader and switches to detection mode. In this mode, PyCRA selects randomized periods of time to issue the physical challenges by switching the RFID antenna from on to off and from off to on.

In order to select an appropriate alarm threshold, we first study the effect of the magnetic coupling between the reader and the tag in the absence of an eavesdropper. This is shown in Figure \ref{fig:RFID_response_pycra_noattack} where the alarm threshold is chosen such that no alarm is triggered when the effect of the magnetic coupling---the residual  between the ``no tag'' response (top) and the response with a tag (middle)---is within the acceptable range. This acceptable residual range allows for coupling induced by the tag only. Any increase on top of this allowable threshold is then attributed to the existence of an excessive magnetic coupling due to an eavesdropper.

Figure \ref{fig:RFID_response_pycra} (middle) shows the response to the same set of physical challenges when the attacker places an eavesdropping antenna in the proximity of the RFID reader while the tag is also present. Thanks to the physical challenge produced by PyCRA, the magnetic coupling produced by the eavesdropper antenna is now easily detected. This can be shown in Figure \ref{fig:RFID_response_pycra} (bottom) which shows the residuals between the expected output and the measured signal exceeding the alarm threshold.

We recorded over 1 hour of RFID measurements with varying distances of the malicious eavesdropping antenna.  Of those experiments where the attacker was close enough to observe the RFID communication, PyCRA successfully detected the existence of an attacker with 100\% accuracy. Intuitively, if the attacker successfully measures the RFID signal, he has consequently removed enough energy from the channel to trigger an alarm.

\begin{figure*}
\hspace{-7.5mm}
\centering
{
	\begin{tabular}{c|c|c}
	\subfloat[ ]{\label{fig:RFID_response_normal}
	\begin{tabular}{c}
		\includegraphics[width=0.274\textwidth]{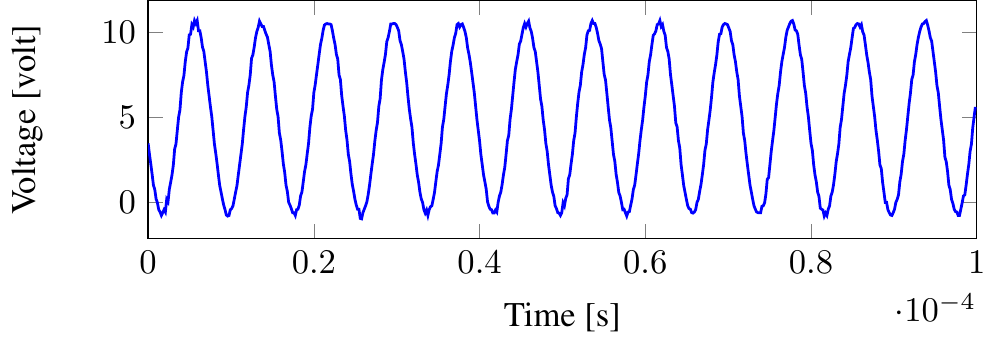}\\
		\includegraphics[width=0.274\textwidth]{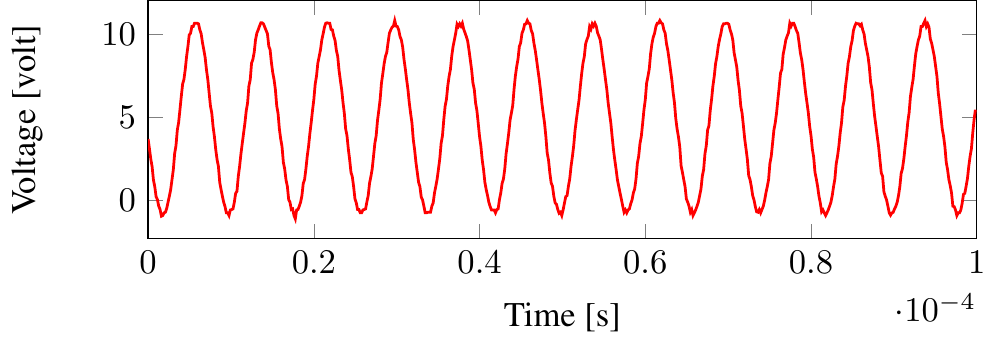}\\
		\includegraphics[width=0.274\textwidth]{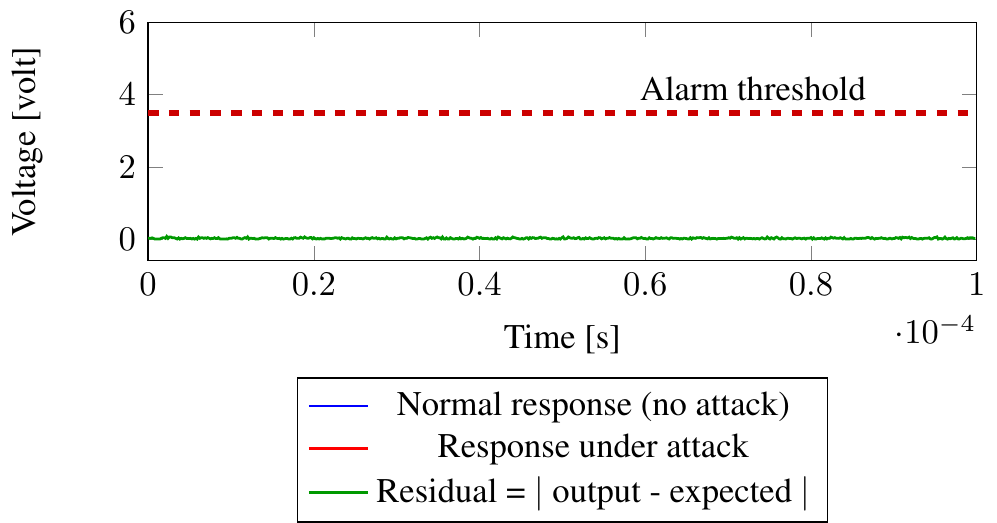}\\
	\end{tabular}
	}&
	\subfloat[ ]{\label{fig:RFID_response_pycra_noattack}
	\begin{tabular}{c}
		\includegraphics[width=0.274\textwidth]{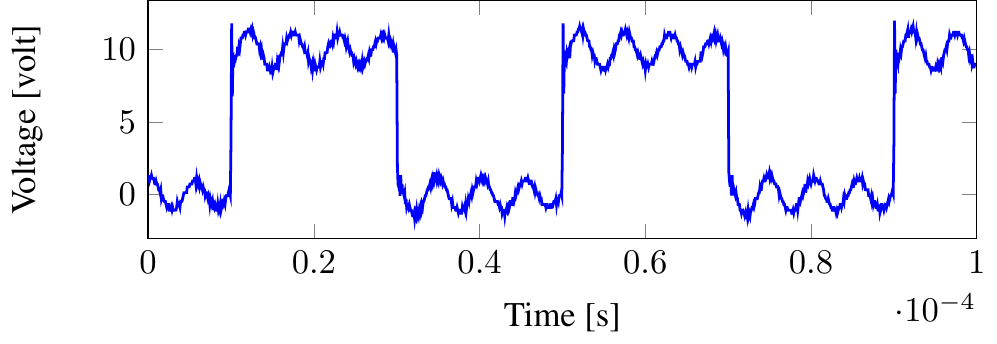}\\
		\includegraphics[width=0.274\textwidth]{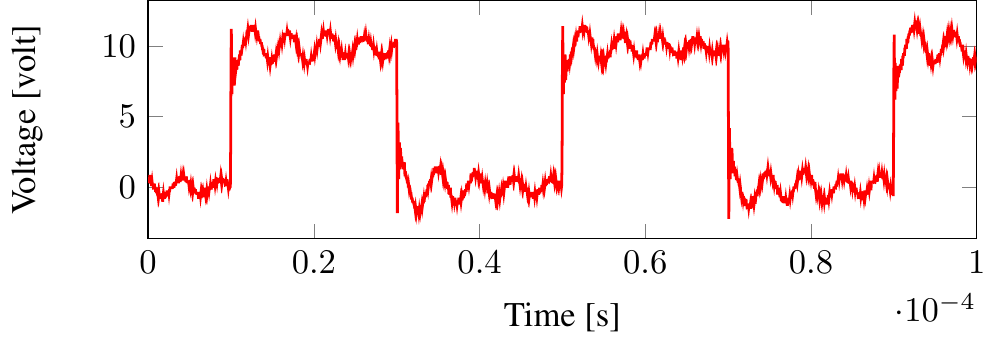}\\
		\includegraphics[width=0.274\textwidth]{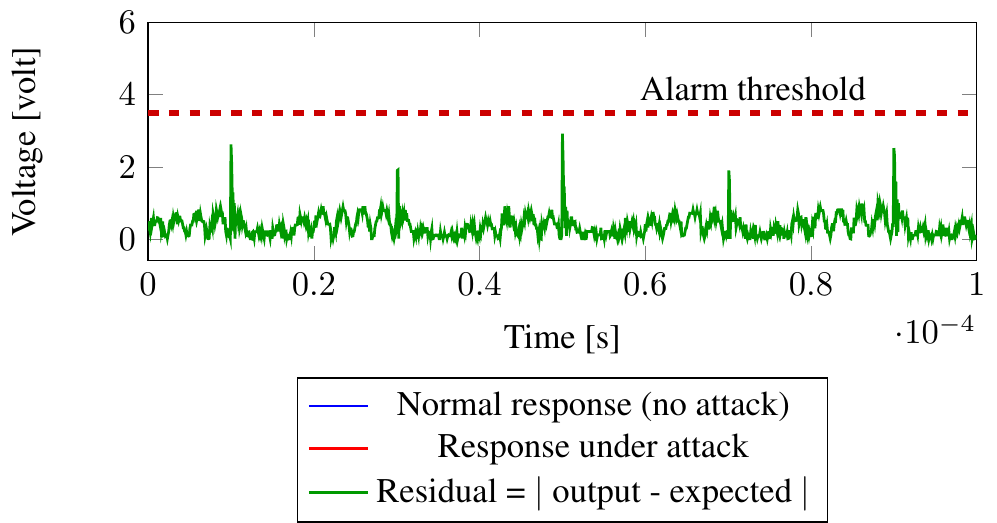}\\
	\end{tabular}
	}&
	\subfloat[ ]{\label{fig:RFID_response_pycra}
	\begin{tabular}{c}
		\includegraphics[width=0.274\textwidth]{Figure_RFID_noTag_noAttack.pdf}\\
		\includegraphics[width=0.274\textwidth]{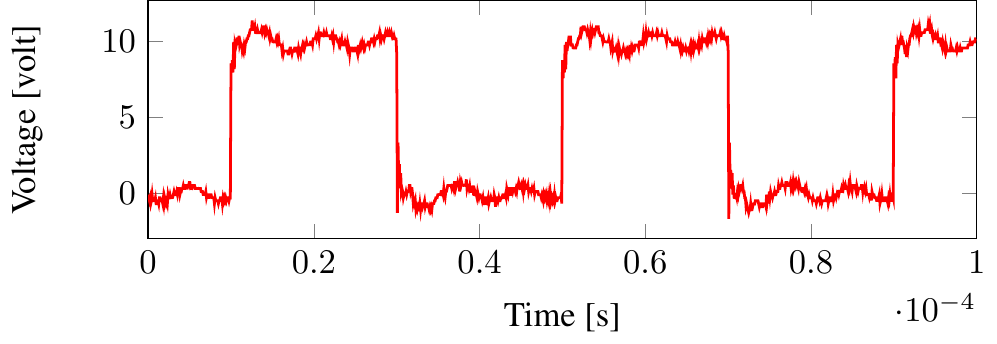}\\
		\includegraphics[width=0.274\textwidth]{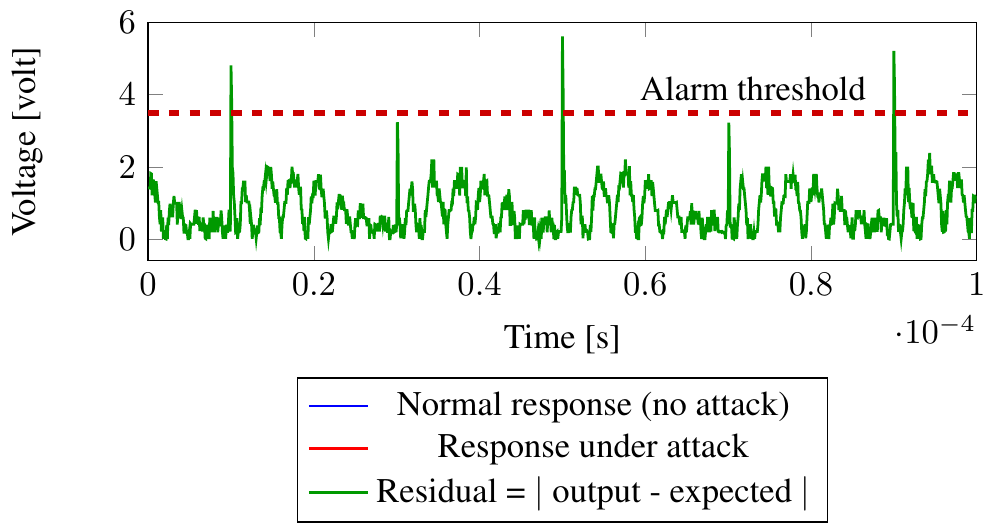}\\
	\end{tabular}
	}
	\end{tabular}
	}
	\caption{\label{fig:RFID_results} Results of applying PyCRA to detect the existence of an eavesdropper in the proximity of an RFID reader. (a) Results of using standard 125KHz signal for detection. (b) Results of using PyCRA when only an RFID tag is present in the proximity of the 
	PyCRA-enabled 
	reader, and (c) Results of using PyCRA in detecting eavesdropping when both an RFID tag and an eavesdropper antenna are present in the proximity of the 
	PyCRA-enabled 
	reader.  Top figure shows the response to the physical challenges when no eavesdropper is placed in the proximity of the RFID reader. The middle figure shows the response to the physical challenges when (a) an eavesdropper antenna (b) passive tag only  (c) passive tag +  eavesdropper antenna are placed in the proximity of the reader. Finally, the bottom figure shows the value of the residuals calculated by PyCRA along with the alarm threshold.}
\vspace{-4mm}
\end{figure*}

\section{Physics, Models, Randomness, \& Security} 
\label{sec:discussion}

The underlying ideas behind PyCRA utilize accurate mathematical models that capture the physics and dynamics of active sensors. With these models in hand, PyCRA is able to isolate the response of the attacker from the response of the environment when challenged by the proposed secure sensor. The success of the system presented in this paper lends credence to the notion of physical challenge-response authentication as a whole.  Section \ref{sec:introduction} made brief mention of several additional sensing modalities where PyCRA can  provide improved security.  In fact, adaptation of the methods described in this work to these sensors is not difficult---one need only revisit the physical models involved in order to derive a method of probing the physical environment and measuring the environmental response.  

The results presented in the previous sections demonstrate the key strengths of PyCRA---namely that it uses fundamental properties of physics (physical delays) along with mathematical models to provided the following security aids (i) timely and accurate detection of external, potentially malicious sources of signal interference, (ii) providing resilience via accurate estimation of malicious spectra, and (iii) accurate detection of passive eavesdropping attacks.

The security guarantees provided by PyCRA depend on two major factors: the underlying physics and the mathematical models. On one hand, physics enforces a fundamental bound on the dynamics of the attacker capabilities and hence enhances the security guarantees (recall results from Figure \ref{fig:detection} showing success rate versus physical delay). As technology evolves, attackers will be able to build better sensors and actuators leading to degraded detection rates, but he or she will never be able to go beyond the limits imposed by physics.
On the other hand, errors in the accuracy of the mathematical model degrade the security guarantees. That is, as the gap between the mathematical model and the underlying sensor physics increases, the attacker gains more room to conceal attacks. For example, scenarios where the physics of the sensor is affected by temperature variations, aging, and so on can increase the gap between the mathematical model and the physics of the sensor. Fortunately, with
 advances in sensor technologies (especially in Micro Electrical Micro Mechanical or MEMS based sensors), better models are developed (at design time) in order to increase the performance of these sensors. Moreover, new technologies are continually developed to allow sensors to adapt over temperature and aging variations. For example, new MEMS technologies in sensors equip the sensor with One-Time-Programmable (OTP) memory used to store calibration parameters versus a wide range of temperature variations \cite{yurish2005smart}. These same calibration parameters can be used by PyCRA in order to adapt its model to overcome temperature variations.

Another major factor in the security provided by PyCRA is the amount of randomness used in generating the physical challenges. The relationship between randomness and security guarantees is a classical relationship that appears in most cryptographic security protocols. However, a fundamental difference in this scheme is that PyCRA relies only on private randomness compared to shared randomness in classical cryptographic authentication schemes. This is a consequence of the ``passivity'' property of the measured entity exploited by PyCRA. This in turn eliminates the classical problem of random data distribution (e.g. key distribution) and thus increases the security provided by the proposed system.


\section{Conclusion}
\label{sec:conclusion}

We have presented PyCRA, a physical challenge-response authentication method for active sensors.  The driving concept  behind PyCRA is that, through random physical stimulation and subsequent behavior analyses, we are able to determine whether or not a sensor is under attack and, in some cases, remain resilient to attacks.  This work further describes ways in which the PyCRA scheme can be applied to (1) passive eavesdropping attacks, (2) simple spoofing attacks, and (3) more advanced spoofing attacks.  We have demonstrated the effectiveness of PyCRA for three case studies: magnetic encoder attack detection, magnetic encoder attack resilience, and RFID eavesdropping detection.  Our results from these case studies indicate that physical challenge-response authentication can accurately and reliably detect and mitigate malicious attacks at the analog sensor level. 

Finally, we believe the results discussed in this work lend credence to the notion of physical challenge-response systems in general, advocating its adoption for active sensors in general where secure operation is a critical component.  More broadly, PyCRA offers security to a wide array of systems (not just sensors) where inputs are mapped to outputs by well-understood models.

\section*{Acknowledgments}
This material is based upon work supported by the NSF under award CNS-1136174. The U.S. Government is authorized to reproduce and distribute reprints for Governmental purposes notwithstanding any copyright notation thereon. The views and conclusions contained herein are those of the authors and should not be interpreted as necessarily representing the official policies or endorsements, either expressed or implied, of NSF or the U.S. Government.

\appendix
\section{Appendix: Proof Of Theorem \ref{th:DelayGuarantee}}\label{sec:proofodDetDelayDecay}
We begin by presenting a discretization argument for applying the change point detection setting to our problem, followed by the main theorem.
\begin{lem}
Assume that PyCRA draws the silent period at random over time according to a probability function over a discrete set of values; for example, assume that $\Pr\pa{\Ga=n\cdot T}=\pa{1-\rho}^{n-1}\cdot \rho$, where $T$ is the length of each interval. In addition, assume that the attacker knows this probability function.
In this case, the problem of finding the probability of a delay which is a multiple of $T$ can be translated to the discrete problem depicted in subsection \ref{subsec:QCPDBasicDefinitionResults} for the Gaussian case. Therefore, the probability of a finite delay which is a multiple of $T$ can be calculated by considering a projection of the observed (continuous) signal onto a discrete set of measurements.
\end{lem}
\begin{proof}
The attacker either receives a random process of AWGN, or the signal reflected from the measured entity (which he knows) along with AWGN. Since the attacker knows the signal, it is statistically sufficient to project the continuous observation (over time) onto the (possibly) reflected signal in each time interval, i.e., projecting in each time interval the observed signal onto a matched filter \cite{ProakisBook2004}. The projection of an AWGN process leads to AWGN observations, and so essentially this is the same problem as the one depicted in subsection \ref{subsec:QCPDBasicDefinitionResults} for the Gaussian case.
\end{proof}

In case $T\ll 1$ it can be assumed that the attacker detects the change over the grid; otherwise, the delay is lower bounded by the highest point in the grid for which the probability of delay is still small enough to fulfill the requirements of PyCRA.
The probability of a delay which is smaller than or equal to $K$ can be written as:
{\small
\beq{}
{
Pr\pa{\Ga\le \tau\le \Ga+K}=\sum_{i=0}^{K}Pr\pa{\tau=\Ga+i}.
}}
For simplicity let us focus on the case where the delay is of length $1$. In this case we are interested in upper bounding the probability 
$Pr\pa{\tau=\Ga+1}$
for any sequential change detector that operates under the constraint
$Pr\pa{\tau<\Ga}\le\Al\ll 1$.
In other words, we are interested in upper bounding the asymptotic decay of $Pr\pa{\tau=\Ga+1}$ as a function of $\Al$, for any sequence of strategies  $\tau\pa{\cdot}\in\Delta\pa{\Al}$.
We begin by taking a closer look at $Pr\pa{\tau=\Ga+1}$:
{\small
\bal{eq:AsDCLDelay1_1}{
Pr \!\!\pa{\tau\!=\!\Ga\!\!+\!\!1}\!=\!\sum_{k=1}^{\infty}Pr\pa{\Ga\!\!=\!\!k,\tau\!=\!k\!+\!1}
\!=\!\sum_{k=1}^{\infty}\!\int\! Pr\!\!\pa{\Ga\!=\!k,\tau\!=\!k\!+\!1|x_{1},\!\dots\!,x_{k\!+\!1}} \!f\!\pa{x_{1},\!\dots\!,x_{k\!+\!1}}dx_{1}\dots dx_{k\!+\!1}
}}
where:
{\small
\bal{eq:SecASGeneralPDF}{
f\pa{x_{1},\dots,x_{k+1}}=\sum_{i=1}^{k+1}\prod_{j=1}^{i-1}f_{0}\pa{x_{j}}
\cdot\prod_{j^{'}=i}^{k+1}f_{1}\pa{x_{j^{'}}}\cdot Pr\pa{\Ga= i}+\prod_{j=1}^{k+1}f_{0}\pa{x_{j}}\cdot \cdot Pr\pa{\Ga> k+1}.
}}
Since $\tau\pa{x_{1},\dots,x_{k+1}}=k+1$ is a function of $x_{1},\dots,x_{k+1}$ we get:
{\small
\bal{eq:AsDCLDelay1_2}{
Pr\pa{\Ga=k,\tau=k+1|x_{1},\dots,x_{k+1}}=&Pr&\pa{\tau=k+1|\Ga=k,x_{1},\dots,x_{k+1}}\cdot
Pr\pa{\Ga=k|x_{1},\dots,x_{k+1}}
\nn\\
&=&\mathbbm{1}\pa{\tau=k+1|x_{1},\dots,x_{k+1}}\cdot Pr\pa{\Ga=k|x_{1},\dots,x_{k+1}}
}}
where $\mathbbm{1}\pa{\cdot}$ is the indicator function.
By assigning back to \eqref{eq:AsDCLDelay1_1} we get:
{\small
\bal{eq:AsDCLDelay1_3}{
Pr\pa{\tau=\Ga+1} \!=\!\! \sum_{k=1}^{\infty}\int Pr\pa{\Ga=k|x_{1},\dots,x_{k+1}} 
\times\mathbbm{1}\pa{\tau=k+1|x_{1},\dots,x_{k+1}}
f\pa{x_{1},\dots,x_{k+1}}dx_{1}\dots dx_{k+1}.
}}
We now wish to establish the relation between $Pr\pa{\Ga=k|x_{1},\dots,x_{k+1}}$ and
$Pr\pa{\Ga>k+1|x_{1},\dots,x_{k+1}}$. We will later use this relation to relate the probability of delay of length one, and the false alarm probability. It can be easily shown that:
{\small
\bal{}{
Pr\pa{\Ga=k|x_{1},\dots,x_{k+1}}=\frac{\pa{1-e^{-1}}^{k-1}e^{-1}}{f\pa{x_{1},\dots,x_{k+1}}}\cdot
\prod_{i=1}^{k-1}f_{0}\pa{x_{i}}\cdot f_{1}\pa{x_{k}}\cdot f_{1}\pa{x_{k+1}}.
}}
whereas $Pr\pa{\Ga>k+1|x_{1},\dots,x_{k+1}}=\frac{\pa{1-e^{-1}}^{k+1}}{f\pa{x_{1},\dots,x_{k+1}}}\cdot
\prod_{i=1}^{k+1}f_{0}\pa{x_{i}}$.
Therefore, we get the relation:
{\small
\bal{eq:AsDCLDelay1FAR}
{
Pr\pa{\Ga=k|x_{1},\dots,x_{k+1}}=Pr\pa{\Ga>k+1|x_{1},\dots,x_{k+1}}\frac{e^{-1}}{\pa{1-e^{-1}}^{2}}\cdot L\pa{x_{k}}\cdot L\pa{x_{k+1}},
}}
where $L\pa{x_{i}}=\frac{f_{1}\pa{x_{i}}}{f_{0}\pa{x_{i}}}=e^{\frac{A}{\sigma^{2}}\pa{x_{i}-A/2}}$.
Now, we look more closely at the probability of false alarm:
{\small
\bal{eq:AsDFADer}{
\alpha=P_{FA}=Pr\pa{\Ga>\tau}=\sum_{i=1}^{\infty}Pr\pa{\Ga>i,\tau=i}
=\sum_{i=1}^{\infty}\int_{x_{1},\dots,x_{i}}\!\!\!\!\!\!\!\!\!\!\!\!\!\!\!Pr\pa{\Ga>i,\tau=i|x_{1},\dots,x_{i}}
f\pa{x_{1},\dots,x_{i}}dx_{1}\dots dx_{i}
}}
where:
{\small
\bal{eq:AsDFADer2}{
Pr\pa{\Ga>i,\tau=i|x_{1},\dots,x_{i}}&=&Pr\pa{\tau=i|\Ga>i,x_{1},\dots,x_{i}}\cdot
Pr\pa{\Ga>i|x_{1},\dots,x_{i}}
\nn\\
&=&\mathbbm{1}\pa{\tau=i|x_{1},\dots,x_{i}}\cdot Pr\pa{\Ga>i|x_{1},\dots,x_{i}}.
}}
By assigning back to \eqref{eq:AsDFADer} we get:
{\small
\bal{eq:AsDFADer3}{
Pr\pa{\Ga>\tau} = \sum_{i=1}^{\infty}\int Pr\pa{\Ga>i|x_{1},\dots,x_{i}}\cdot \mathbbm{1}\pa{\tau=i|x_{1},\dots,x_{i}}\cdot
f\pa{x_{1},\dots,x_{i}}dx_{1}\dots dx_{i}.
}}
The next lemma proves the relation between element $k$ in \eqref{eq:AsDCLDelay1_3} and element $k+1$ in \eqref{eq:AsDFADer3}.
\begin{lem}\label{lem:EachElementAsymptiticEq}
{\small
\bal{eq:AsDLemEAFAREL}
{
&\int Pr\pa{\Ga=k|x_{1},\dots,x_{k+1}}\cdot
\mathbbm{1}\pa{\tau=k+1|x_{1},\dots,x_{k+1}}\cdot f\pa{x_{1},\!\dots\!,x_{k+1}}dx_{1}\dots dx_{k+1}=
\nn\\
&\int\!\! Pr\pa{\Ga\!>\!k\!+\!1|x_{1},\!\dots\!,x_{k\!+\!1}}\!\!\cdot\!\! \mathbbm{1}\pa{\!\tau\!=\!k\!+\!1|x_{1},\dots,x_{k+1}}\!\!\times\!\!
\frac{e^{-1}}{\pa{1-e^{-1}}^{2}} \!\!\cdot\!\! L\pa{x_{k}}
L\pa{x_{k\!+\!1}}\!\cdot\! f\pa{x_{1},\dots,x_{k+1}}dx_{1}\!\dots\! dx_{k\!+\!1}
}}
Further, for $\Al\ll 1$ we get:
{\small
\bal{eq:AsDLemEAFAREL2}{
\int Pr\pa{\Ga=k|x_{1},\dots,x_{k+1}}\cdot
\mathbbm{1}\pa{\tau=k+1|x_{1},\dots,x_{k+1}}\cdot f\pa{x_{1},\dots,x_{k+1}}dx_{1}\dots dx_{k+1}\dot{=}
\nn\\
\int Pr\pa{\Ga>k+1|x_{1},\dots,x_{k+1}}\cdot\mathbbm{1}\pa{\tau=k+1|x_{1},\!\dots\!,x_{k+1}}\cdot f\pa{x_{1},\!\dots\!,x_{k+1}}dx_{1}\dots dx_{k+1}.
}}
\end{lem}
\begin{proof}
the equality in \eqref{eq:AsDLemEAFAREL} is obtained by assigning \eqref{eq:AsDCLDelay1FAR} to \eqref{eq:AsDCLDelay1_3}.
We prove the asymptotic equality in \eqref{eq:AsDLemEAFAREL2} by considering the behavior as $\Al\to 0$. Let us consider the following representation for the random variables:
{\small
\beq{}{
x_{i}=a_{i}\cdot \sqrt{\ln\pa{1/\alpha}}\qquad 1\le i\le k+1.
}}
When assigning this representation to $L\pa{x_{i}}$, we get $L\pa{a_{i}}=e^{\frac{A}{\sigma^{2}}\pa{a_{i}\sqrt{\ln\pa{1/\alpha}}-A}},$
whereas assigning this representation to the PDFs yields $f_{0}\pa{a_{i}}=\frac{1}{\sqrt{2\cdot \sigma^{2}}}\cdot e^{-a_{i}^{2}\ln\pa{\frac{1}{\Al}}/{2\sigma^{2}}}$ and $f_{1}\pa{a_{i}}=\frac{1}{\sqrt{2\cdot \sigma^{2}}}\cdot e^{-\pa{a_{i}\sqrt{\ln\pa{1/\Al}}-A}^{2}/{2\sigma^{2}}}$.
Note that $L\pa{a_{i}}\cdot f_{0}\pa{a_{i}}\dot{=} f_{0}\pa{a_{i}}$ and $L\pa{a_{i}}\cdot f_{1}\pa{a_{i}}\dot{=} f_{1}\pa{a_{i}}$. Therefore, based on the definition of $f\pa{x_{1},\dots,x_{k}}$ in equation \eqref{eq:SecASGeneralPDF}, we get:
{\small
\beq{eq:ExpEqBetweenPDFs}
{
L\pa{a_{k}}\times\cdots\times L\pa{a_{k+K}}\cdot f\pa{a_{1},\dots,a_{k+K}}\dot{=}f\pa{a_{1},\dots,a_{k+K}}\quad\forall k\ge 0.
}}
We now prove the asymptotic equality in \eqref{eq:AsDLemEAFAREL2}:
{\small
\bal{}
{
&\ln^{\frac{k+1}{2}}\pa{{\frac{1}{\Al}}}\int Pr\pa{\Ga=k|a_{1},\dots,a_{k+1}}\cdot
\mathbbm{1}\pa{\tau=k+1|a_{1},\dots,a_{k+1}}\cdot f\pa{a_{1},\dots,a_{k+1}}da_{1}\dots da_{k+1}=
\nn\\
&\ln^{\frac{k+1}{2}}\pa{{\frac{1}{\Al}}}\!\!\int\!\! Pr\pa{\Ga>k+1|a_{1},\dots,a_{k+1}}\!\!\cdot\!\!\mathbbm{1}\pa{\tau=k+1|a_{1},\dots,a_{k+1}} \!\!\times\!\!
L\pa{a_{k}}\!\!\cdot\!\! L\pa{a_{k+1}}\!\!\cdot\!\! f\pa{a_{1},\dots,a_{k+1}}da_{1}\dots da_{k+1}. \nn
}}
From \eqref{eq:ExpEqBetweenPDFs} we get:
{\small
\bal{}{
&\ln^{\frac{k+1}{2}}\pa{{\frac{1}{\Al}}}\!\!\int\!\! Pr\pa{\Ga\!>\!k\!+\!1|a_{1},\dots,a_{k\!+\!1}}\cdot\mathbbm{1}\pa{\tau\!=\!k\!+\!1|a_{1},\dots,a_{k+1}} \!\times\!
L\pa{a_{k}}\cdot L\pa{a_{k+1}}\cdot f\pa{a_{1},\dots,a_{k+1}}da_{1}\dots da_{k+1}\dot{=}
\nn\\
&\ln^{\frac{k+1}{2}}\pa{{\frac{1}{\Al}}}\int Pr\pa{\Ga>k+1|a_{1},\dots,a_{k+1}} \cdot \mathbbm{1}\pa{\tau=k+1|a_{1},\dots,a_{k+1}}
\times f\pa{a_{1},\dots,a_{k+1}}da_{1}\dots da_{k+1} \nn
}}
which proves the equality in \eqref{eq:AsDLemEAFAREL2}\footnote{A more accurate argument for the asymptotic equality can be obtained by partitioning the range of $a_{k}, a_{k+1}$ into two sets: $|a|\le\ln^{1/2-\epsilon}\pa{1/\Al}$ and $|a|>\ln^{1/2-\epsilon}\pa{1/\Al}$; for the former we consider the integrand that has the minimal exponent in the range, whereas for the latter we consider only $L\pa{a_{k}}\cdot L\pa{a_{k+1}}\cdot f\pa{a_{k},a_{k+1}}$.}.
\end{proof}

Now, based on the result from Lemma \ref{lem:EachElementAsymptiticEq} we show that $Pr\pa{\tau=\Ga+1}\dot{=}\Pr\pa{\Ga>\tau}=\Al$ for any $\tau\in\Delta\pa{\Al}$
\begin{lem}\label{lem:Convinfsumtofalsealarm}
For any sequence of strategies (as a function of $\Al\ll 1$) $\tau\pa{\cdot}\in\Delta\pa{\Al}$
{\small
\beq{}{
Pr\pa{\tau=\Ga+1}\dot{=}\Pr\pa{\Ga>\tau}=\Al
}}
\end{lem}
\begin{proof}
Element $k$ in \eqref{eq:AsDCLDelay1_3} is lower bounded by $Pr\pa{\Ga=k}$. Hence, for $k_{0}=\ln\pa{\frac{1}{\Al}}/\ln\pa{1-\rho}$, we get:
{\small
\bal{}
{
\sum_{k=k_{0}}^{\infty}\!\!\!\int \!\!\!Pr\pa{\Ga=k|x_{1},\dots,x_{k+1}} \cdot \mathbbm{1}\pa{\tau=k+1|x_{1},\dots,x_{k+1}}
\times f\pa{x_{1},\dots,x_{k+1}}dx_{1}\dots dx_{k+1} \le  \sum_{k=k_{0}}^{\infty}\!\!Pr\pa{\Ga=k}=\Al. \nn
}}
In addition:
{\small
\bal{}
{
\sum_{k=1}^{k_{0}-1}\int Pr\pa{\Ga=k|x_{1},\dots,x_{k+1}}\cdot\mathbbm{1}\pa{\tau=k+1|x_{1},\dots,x_{k+1}}\cdot f\pa{x_{1},\dots,x_{k+1}}dx_{1}\dots dx_{k+1}\le
\nn\\
\frac{\ln\pa{1/\Al}}{\ln\pa{1-\rho}}\max_{1\le k\le k_{0}-1}\int Pr\pa{\Ga=k|x_{1},\dots,x_{k+1}}\cdot\mathbbm{1}\pa{\tau=k+1|x_{1},\dots,x_{k+1}}
\times f\pa{x_{1},\dots,x_{k+1}}dx_{1}\dots dx_{k+1} \nn
}}
Furthermore, From Lemma \ref{lem:EachElementAsymptiticEq} we get that:
{\small
\bal{}
{
\max_{1\le k< k_{0}}
\int Pr\pa{\Ga=k|x_{1},\dots,x_{k+1}}\cdot\mathbbm{1}\pa{\tau=k+1|x_{1},\dots,x_{k+1}}\cdot f\pa{x_{1},\dots,x_{k+1}}dx_{1}\dots dx_{k+1}\dot{=}
\nn\\
\max_{1\le k< k_{0}} \int Pr\pa{\Ga>k+1|x_{1},\dots,x_{k+1}}\cdot\mathbbm{1}\pa{\tau=k+1|x_{1},\dots,x_{k+1}}\times
f\pa{x_{1},\dots,x_{k+1}}dx_{1}\dots dx_{k+1}\le \Al \nn
}}
where the first equality means that the maximal integrands have the same decay rate as a function of $\Al$. The equality for the maximal integrands holds even though $k_{0}$ increases with $\Al$; this is due to the fact that the elements in \eqref{eq:AsDLemEAFAREL2} converge to equality uniformly over $k$. They converge uniformly since the integrands are affected by the same functions for any $k$, namely, $L\pa{\cdot}$, $f_{0}\pa{\cdot}$ and $f_{1}\pa{\cdot}$, and also because the range in which the maximal integrand resides is the same for any $k$, i.e., $|a|\le\ln^{1/2-\eps}\pa{1/\Al}$.

Combining these two arguments concludes the proof.
\end{proof}

\begin{cor}
{\small
\beqn{}
{
Pr\pa{\Ga\le \tau\le\Ga+K}\dot{=}\Al
}}
\end{cor}

\begin{proof}
When $K$ is a constant that does not depend on $\Al$, the proof for $Pr\pa{\tau=\Ga+i}\dot{=}\Al$ for $0\le i\le K$, follows along the same lines as Lemma \ref{lem:EachElementAsymptiticEq} and Lemma\ref{lem:Convinfsumtofalsealarm}. Therefore:
{\small
\beqn{}
{
Pr\pa{\Ga\le \tau\le\Ga+K}=\sum_{i=0}^{K} Pr\pa{\tau=\Ga+i}\dot{=}\Al
}}
which concludes the proof.
\end{proof}

\bibliographystyle{ACM-Reference-Format-Journals}
\bibliography{ICCPS2013}



\medskip

\end{document}